\documentclass[11pt]{article}
\usepackage{amsfonts,amsmath,amsthm,amssymb, etex, mathtools,xcolor}
\usepackage[T1]{fontenc}
\usepackage{tgtermes}
\usepackage[pagebackref]{hyperref}
\hypersetup{
colorlinks=true,
urlcolor=blue,
linkcolor=blue,
citecolor=blue,
}
\usepackage{paralist}
\usepackage{enumerate}
\usepackage{framed}
\usepackage{tikz}
\usetikzlibrary{arrows,automata}
\usepackage{comment}
\usepackage[margin=1in]{geometry}
\usepackage{graphicx,float,wrapfig}
\usepackage{mathrsfs}
\usepackage[usenames,dvipsnames]{pstricks}
\usepackage{epsfig}
\usepackage{pst-grad} 
\usepackage{pst-plot} 
\usepackage[boxed,ruled]{algorithm}
\usepackage[noend]{algorithmic}

\usepackage{bbm}
\usepackage{tabu}
\usepackage{xfrac}
\usepackage[normalem]{ulem}

\newtheorem{theorem}{Theorem}[section]

\newtheorem{lemma}[theorem]{Lemma}

\newtheorem{cor}[theorem]{Corollary}
\theoremstyle{definition}
\newtheorem{definition}[theorem]{Definition}
\theoremstyle{plain}
\newtheorem{remark}[theorem]{Remark}

\newenvironment{proofof}[1]{\begin{proof}[{\bf Proof of #1}]}{\end{proof}}

\newenvironment{reminder}[1]{\bigskip
	\noindent {\bf Reminder of #1  }\em}{\smallskip}

\def\ShowAuthNotes{1}
\ifnum\ShowAuthNotes=1
\newcommand{\authnote}[2]{\ \\ \textcolor{red}{\parbox{0.9\linewidth}{[{\footnotesize {\bf #1:} { {#2}}}]}}\newline}
\else
\newcommand{\authnote}[2]{}
\fi

\DeclarePairedDelimiter\abs{\lvert}{\rvert}%

\newcommand{\cnote}[1]{\authnote{Ce}{#1}}

\renewcommand{\epsilon}{\varepsilon}
\newcommand{\Ex}{\operatornamewithlimits{\mathbb{E}}}
\newcommand{\F}{\mathbb{F}}
\renewcommand{\tilde}{\widetilde}
\newcommand{\poly}{\operatorname{\mathrm{poly}}}
\newcommand{\polylog}{\poly\log}
\newcommand{\eps}{\epsilon}

\newcommand{\N}{\mathbb{N}}
\newcommand{\Z}{\mathbb{Z}}

\def \C {\mathbb{C}}

\def\caH{\mathcal{H}}

\usepackage{xspace}
\def \SS {Subset Sum\xspace}

\DeclarePairedDelimiter\floor{\lfloor}{\rfloor}

\begin{document}

\title{Fast Low-Space Algorithms for Subset Sum\thanks{Supported by NSF CCF-1741615 and NSF CCF-1909429.}}
\author{Ce Jin\\MIT\\ \texttt{cejin@mit.edu} \and Nikhil Vyas \\ MIT \\ \texttt{nikhilv@mit.edu} \and Ryan Williams\\MIT\\\texttt{rrw@mit.edu}}

\maketitle
\begin{abstract}
We consider the canonical Subset Sum problem: 
given a list of positive integers $a_1,\ldots,a_n$ and a target integer $t$ with $t > a_i$ for all $i$, determine if there is an $S \subseteq [n]$ such that $\sum_{i \in S} a_i = t$.
The well-known pseudopolynomial-time dynamic programming algorithm [Bellman, 1957] solves Subset Sum in $O(nt)$ time, while requiring $\Omega(t)$ space. 

In this paper we present algorithms for Subset Sum with $\tilde O(nt)$ running time and much lower space requirements than Bellman's algorithm, as well as that of prior work. We show that Subset Sum can be solved in $\tilde O(nt)$ time and $O(\log(nt))$ space with access to $O(\log n \log \log n+\log t)$ random bits. This significantly improves upon the $\tilde O(n t^{1+\varepsilon})$-time, $\tilde O(n\log t)$-space algorithm of Bringmann (SODA~2017).
We also give an $\tilde O(n^{1+\varepsilon}t)$-time, $O(\log(nt))$-space randomized algorithm,  improving upon previous $(nt)^{O(1)}$-time $O(\log(nt))$-space algorithms by Elberfeld, Jakoby, and Tantau (FOCS~2010), and Kane (2010).
In addition, we also give a $\mathrm{poly} \log(nt)$-space, $\tilde O(n^2 t)$-time deterministic algorithm.

We also study time-space trade-offs for Subset Sum. For parameter $1\le k\le \min\{n,t\}$, we present a  randomized algorithm running in $\tilde O((n+t)\cdot k)$ time and $O((t/k)\polylog (nt))$ space.

As an application of our results, we give an $\tilde{O}(\min\{n^2/\eps, n/\eps^2\})$-time and $\polylog(nt)$-space algorithm for ``weak'' $\eps$-approximations of Subset Sum.
\end{abstract}

\thispagestyle{empty}
\addtocounter{page}{-1}
\newpage    

\section{Introduction}
We consider the classical \SS problem in its standard form: given positive integers $a_1,\ldots,a_n$ and a target integer $t$ with $t > a_i$ for all $i$, determine if there is an $S \subseteq [n]$ such that $\sum_{i \in S} a_i = t$.
In the 1950s, Bellman~\cite{Bellman57} showed that \SS is in $O(n  t)$ time, with a textbook \emph{pseudopolynomial} time algorithm.\footnote{In this paper, we work in a random-access model of word length $\Theta(\log n + \log t)$, and measure space complexity in total number of bits. } The algorithm is also a textbook example of \emph{dynamic programming}, requiring $\Omega(t)$ space to store the table. In this paper, we address the question: to what extent can the space complexity of \SS be reduced, without increasing the running time of Bellman's algorithm?\footnote{Recall the ``space complexity'' of an algorithm is the size of the working memory used by it; the input integers $a_1,\dots,a_n$ are assumed to be stored in a read-only randomly-accessible array, which does not count towards the space complexity. } Besides the inherent interest in designing algorithms with tiny space overhead, recall that low-space algorithms also imply efficient \emph{parallel algorithms} due to the well-known result that space-$s$ algorithms can be solved in $O(s^2)$ parallel time with $\poly(2^s)$ work/processors~\cite[Theorem 16.1]{Papadimitriou94}. For example, an $O(\log (nt))$-space algorithm can be applied to solve Subset Sum in $O(\log^2(nt))$ parallel time with $\poly(nt)$ work/processors.

\paragraph{Prior Work.} Within the last decade, there has been substantial work on improving the space complexity of pseudopolynomial-time algorithms for \SS and related problems. Lokshtanov and Nederlof~\cite{LokshtanovN10} gave an algorithm running in $\tilde{O}(n^3 t)$ time and $\tilde{O}(n^2\log t)$ space, using Fast Fourier Transforms over $\C$.
Elberfeld, Jakoby, and Tantau~\cite{ElberfeldJT10} gave a generic meta-theorem for placing problems in LOGSPACE, and their result implies that \SS is in pseudopolynomial time and logarithmic space. However no nice bounds on the running time follow (for example, they need Reingold's algorithm for \mbox{s-t} connectivity~\cite{Reingold08}, which has a rather high-degree polynomial running time). Kane~\cite{Kane10} 
gave a more explicit logspace algorithm; it can be shown that his algorithm as stated solves \SS  in deterministic $\tilde O(n^3 t+n^{2.025}t^{1.025})$ time and $O(\log (nt))$ space.\footnote{Kane does not give an explicit running time analysis, but the best known (unconditional) result in number theory on gaps between primes (namely, that there is always a prime in the interval $[N,N+O(N^{0.525})]$~\cite{baker2001difference}) 
implies such a running time. See Remark~\ref{rem:kane}.} 
Recently, Bringmann~\cite{karl} showed that, \emph{assuming the Generalized Riemann Hypothesis}, the problem can be solved in randomized $\tilde{O}(nt)$ time and $\tilde{O}(n \log t)$ space, and unconditionally in randomized $\tilde{O}(nt^{1+\eps})$ time and $\tilde{O}(nt^{\eps})$ space. Although the running time of Bringmann's algorithm is very close to the dynamic programming solution, the space bound is still $\Omega(n)$. (Bringmann also gives an $\tilde{O}(n+t)$ time algorithm for \SS that uses $\Omega (t)$ space.)

\subsection{Our Results}

We extend the algebraic and randomized approaches to \SS in novel ways, obtaining randomized algorithms that use very few random bits and essentially preserve Bellman's running time, while reducing the space usage all the way to $O(\log(n t))$. We also obtain deterministic algorithms. Our first main result is the following.

\begin{theorem} \label{thm:main}
The Subset Sum problem can be solved by randomized algorithms  running in
\begin{enumerate}
 \item $\tilde O(nt)$ time and
 $O(\log(nt))$
 space, and read-only random access to $O(\log n \log \log n )$ random bits.
 \item $\tilde O(n^{1+\eps}t)$ time and $O(\log(nt))$ space, for any $\eps>0$, with 
 $O(\log n)$
  random bits.
\end{enumerate}
\end{theorem}

Our algorithms are Monte Carlo, in that they always answer NO if there is no solution, and answer YES with probability at least $99\%$ if there is a solution. We also obtain a deterministic algorithm with polylogarithmic space and $\tilde O(n^2 t)$ time.

\begin{theorem}
 \label{thm:main_det}
\SS can be solved deterministically in $\tilde O(n^2 t)$ time and $O(\log t\cdot \log^3 n)$ space.
\end{theorem}

Furthermore, we obtain time-space trade-off algorithms for \SS.

\begin{theorem}
 \label{thm:main_tradeoff}
For every parameter $1\le k\le \min\{n,t\}$, there is a randomized algorithm for Subset Sum with $0.01$ one-sided error probability, running in $\tilde O((n+t)\cdot k)$ time and $O((t/k)\polylog (nt))$ space.
\end{theorem}

For $k=1$, Theorem~\ref{thm:main_tradeoff} matches the near-linear time algorithm by Bringmann~\cite{karl}. For larger $k$, we obtain the first algorithm that uses $o(t)$ space and runs faster than Bellman's dynamic programming.

It is interesting to compare the time-space tradeoff of Theorem~\ref{thm:main_tradeoff} with known conditional lower bounds on \SS. Note that the product of the time bound $T$ and space bound $S$ in Theorem~\ref{thm:main_tradeoff} is always at least $\tilde{\Theta}(nt + t^2)$.
Results in fine-grained complexity show that, when $T=S$ (the time bound equals the space bound), an algorithm running in $t^{1-\eps/2} 2^{o(n)}$ time and $t^{1-\eps/2} 2^{o(n)}$ space for \SS would refute SETH~\cite{abboud2019seth}, so in such a case the time-space product must be at least $t^{2-\eps} \cdot 2^{o(n)}$. It is interesting that, while \SS can be solved in essentially $nt$ time and logarithmic space (Theorem~\ref{thm:main}) and also in $\tilde{O}((n+t)\cdot k)$ time and $\tilde{O}(t/k)$ space, obtaining $t^{.99} \cdot 2^{o(n)}$ time and $t^{.99} \cdot 2^{o(n)}$ space appears to be impossible!


\begin{remark} 
\label{rem:range}
Our algorithms solve the decision version of \SS. Another well-studied setting is the optimization version, where we want to find the largest $t'\le t$ that can be expressed as a subset sum $t'=\sum_{i \in S}a_i$.  We remark that the optimization version reduces to the decision version, with an ${O}(\log t)$ extra factor in time complexity: perform binary search for $t'$, calling an oracle that determines if there is a subset sum with value $x\in [t',t]$ for our current $t'$. To implement this oracle, we add $\lfloor \log (t-t'+1) \rfloor +1$ numbers $1,2,4,8,\dots,2^{\lfloor \log(t-t'+1) \rfloor-1}, (t-t') - 2^{\lfloor \log(t-t'+1)\rfloor }+1$ into the input multiset. The subset sums of these extra numbers are exactly $0,1,2,\dots,t-t'$, so we can implement the oracle by running the decision algorithm with target number $t$. 

Note that a space-$s$ decision algorithm for Subset Sum also implies a search algorithm that finds a subset sum solution
(assuming the outputs of the algorithm are written in an append-only storage), with space complexity $O(s+\log t)$ and an extra multiplicative factor of $n$ in the time complexity. By iterating over all $i=1,\ldots,n$, and solving the instance $a_{i+1},\ldots,a_n$ with target $t-a_i$ (which can be stored in $O(\log t)$ space), we can determine whether $a_i$ can be included in the subset sum solution for all $i$, updating the target value accordingly. (For each $a_i$, we can output whether or not it is included on append-only storage.)
\end{remark}

The prior results (compared with ours) are summarized in Figure~\ref{fig:table}.
\begin{figure*}
    \centering
\begin{tabular}{|c|c|c|c|c|}
\hline \hline
{\bf Reference} & {\bf Time} & {\bf Space (bits)} &  {\bf D/R} & {\bf Caveats}\\
\hline
\hline
 \cite{LokshtanovN10} & $\tilde{O}(n^3 t)$ & $\tilde{O}(n^2\log t)$ & D & \rule{-0.55ex}{2.6ex} \\
\cite{ElberfeldJT10} & $(nt)^{O(1)}$ & $O(\log(nt))$ & D & \\
\cite{Kane10} & 
$\tilde O(n^3 t+n^{2.025}t^{1.025})$
& $O(\log (nt))$ & D & \\
Corollary~\ref{cor:kane} & $\tilde{O}(n^3 t)$ & $O(\log (nt))$ & D & \\
Theorem~\ref{thm:main_det}& $\tilde O(n^2 t)$ &$O(\log t\cdot  \log^3 n) $ & D & \\
\hline
\cite{karl} & $\tilde{O}(nt)$ & $\tilde{O}(n \log t)$ & R &  \rule{-0.55ex}{2.6ex}{\bf Requires GRH}\\
\cite{karl} & $\tilde{O}(nt^{1+\eps})$ & $\tilde{O}(nt^{\eps})$ & R & \\
Theorem~\ref{thm:main}& $\tilde O(nt)$ & $O(\log(nt))$ & R & $\begin{matrix}\textbf{Access to}\vspace{-0.12cm}\\O(\log n\log \log n)\vspace{-0.10cm}\\\textbf{random bits}\end{matrix}$\\
Theorem~\ref{thm:main}& $\tilde O(n^{1+\eps}t)$ & $O(\log(nt))$ & R & \\
\hline
\cite{karl} & $\tilde{O}(n+t)$ & $\tilde{O}( t)$ & R & \rule{-0.55ex}{2.6ex} \\
Theorem~\ref{thm:main_tradeoff}& $\tilde{O}((n+t)\cdot k)$ & $\tilde{O}( t/k)$ & R & $1\le k \le \min\{n,t\}$ \\
\hline \hline
\end{tabular}
    \caption{Best known (pseudo-polynomial) time-space upper bounds for solving Subset Sum with $n$ elements with a target $t$. ``D'' means the algorithm is deterministic,  ``R'' means randomized.}
    \label{fig:table}
\end{figure*}

\paragraph{An Approximation Algorithm.} As an application of our techniques, we also give ``weak'' approximation algorithms for subset sum with low space usage. We define a \emph{weak $\eps$-approximation algorithm} to be one that, on an instance $A = [a_1, a_2, \ldots, a_n], t$, can distinguish the following two cases:
\begin{enumerate}
    \item {\bf YES:} There is a subset $S \subseteq [n]$ such that $(1-\eps/2) t\le \sum_{i \in S} a_i \le t$.
    \item {\bf NO:} For all subsets $S \subseteq [n]$ either $\sum_{i \in S} a_i > (1+\eps)t$ or $\sum_{i \in S} a_i < (1-\eps)t$.
\end{enumerate}
Mucha, W\k{e}grzycki, and W{\l}odarczyk~\cite{MuchaW019}, who introduced the search version of the above problem, gave a ``weak'' $\eps$-approximation algorithm in $\tilde{O}(n+1/\eps^{5/3})$ time and space. They observed that this weak approximation algorithm implies an  approximation algorithm (in the usual sense) for the \emph{Partition problem} (see \cite{Bringmann-approx} for further improvement  on this problem), which is a special case of the Subset Sum problem where the input satisfies $a_1+\dots+a_n = 2t$. 

We show that our techniques can be applied to give an extremely space-efficient algorithm for weak approximations of \SS:
\begin{theorem}\label{thm:approx}
There exists a $\tilde{O}(\min\{n^2/\eps, n/\eps^2\})$-time and $O(\polylog(nt))$-space algorithm for weak approximation of \SS.
\end{theorem}
\subsection{Overview}

All of our algorithms build upon a key idea in Kane's logspace algorithm. Kane skillfully applies the following simple number-theoretic fact (the proof of which we will recall in Section~\ref{sec:kane}):
\begin{lemma}
\label{lem:num-theo} Let $p$ be an odd prime and let $a \in [1,2p-3]$ be an integer.
\begin{itemize}
\item If $a \neq p-1$ then $\sum\limits_{x=1}^{p-1}x^a \equiv 0 \pmod{p}$. 
\item If $a=p-1$ then $\sum\limits_{x=1}^{p-1}x^a \equiv p-1 \pmod{p}$.
\end{itemize}
\end{lemma}

Given an instance $a_1,\ldots,a_n,t$ of \SS, Kane constructs the ``generating function'' $E(x) = x^{p-1-t}\prod_{i=1}^n (1+x^{a_i}) \pmod{p}$ for a prime $p > nt$. In particular, we compute $E' = \sum_{x=1}^{p-1} E(x) \pmod{p}$. The key idea is that, by Lemma~\ref{lem:num-theo}, subset sums equal to $t$ will contribute $p-1$ in the sum $E'$, while all other subset sums will contribute $0$, allowing us to detect subset sums (if $p$ is chosen carefully). The running time for evaluating $E'$ is $\Omega(n^2 t)$ because of the large choice of $p$. Intuitively, Kane needs $p \geq \Omega(nt)$ because the ``bad'' subset sums (which we want to cancel out) can be as large as $\Omega(nt)$, and modulo $p$ there is no difference between $t$ and $t+p$.

Our algorithms manage to get away with choosing $p \leq \tilde O(t)$, but they accomplish this in rather different ways. We now briefly describe the three main theorems.

\paragraph{The Randomized Logspace Algorithm.} Our randomized logspace algorithms (Theorem~\ref{thm:main}) combine ideas from Kane's logspace algorithm and Bringmann's faster $\tilde{O}(n+t)$-time algorithm, applying several tools from the theory of pseudorandom generators ($k$-wise $\delta$-dependent hash functions and explicit constructions of expanders) to keep the amount of randomness low.

The key is to apply an idea from Bringmann's \SS algorithm~\cite{karl}: we randomly separate the numbers of the \SS instance into different ranges. For each range, we compute corresponding generating functions on the numbers in those ranges, and combine ranges (by multiplication) to obtain the value of the generating function for the whole set. If we can keep the individual ranges small, we can ensure a prime $p \leq \tilde O(t)$ suffices. To perform the random partitioning with low space, we define a new notion of ``efficiently invertible hash functions'' and prove they can be constructed with good parameters. 
Such a hash function $h$ has a load-balancing guarantee similar to a truly random partition, and the additional property that, given a hash value $v$, we can list all the preimages $x$ satisfying $h(x)=v$ with time complexity that is near-linear in the number of such preimages. Such families of invertible hash functions may be of independent interest.

\paragraph{The Deterministic Algorithm.} In the randomized logspace algorithm, we modified Kane's algorithm so it can work modulo a prime $p \leq \tilde{O}(t)$. This was achieved by mimicking Bringmann's randomized partitioning and color-coding ideas. While this approach yields fast low-space randomized algorithms, it seems one cannot derandomize color-coding deterministically without a large blowup in time. 

Instead of randomly partitioning, our deterministic algorithm (Theorem~\ref{thm:main_det}) uses a very different approach. We apply deterministic approximate counting in low space (approximate logarithm), in the vein of Morris's algorithm~\cite{morris1978counting} for small-space approximate counting, to keep track of the number of elements contributing to a subset sum. This allows us to work modulo a prime $p \leq \tilde{O}(t)$. Along the way, we define a special polynomial product operation that helps us compute and catalog approximate counts efficiently.

\paragraph{The Time-Space Trade-off.} Our time-space trade-off algorithm (Theorem~\ref{thm:main_tradeoff}) uses a batch evaluation idea and additional algebraic tricks. In Kane's framework, we need to evaluate the generating function at all non-zero points $x=1,2,\dots,p-1$. To reduce the running time, we perform the evaluation in $k$ batches, where each batch has $(p-1)/k$ points to evaluate.\footnote{For simplicity, in this discussion we work modulo a prime $p$. For technical reasons, our actual algorithms have to work over an arbitrary finite field $\F_q$ where $q$ is a prime power. Informally, this is because we need our parameter $k$ to divide $p-1$; if we allow $p$ to be an arbitrary prime power and apply some analytic number theory, this is essentially possible by adjusting $p$ and $k$.} In one batch, letting the evaluation points be $b_1,b_2,\dots, b_{(p-1)/k}$, we define a polynomial \[B(x) := (x-b_1)(x-b_2)\cdots(x-b_{(p-1)/k})\] of degree $(p-1)/k$. By choosing the batches of evaluation points judiciously (see Lemma~\ref{lemma:cool}), we can ensure that for every batch, the polynomial $B(x)$ has only two terms, making it possible to perform the $\bmod\  B(x)$ operation efficiently in low space.
We compute the generating function modulo $B(x)$, which agrees with the original polynomial on the evaluation points $b_i$. Then, we can efficiently recover the values at these points using fast multipoint evaluation~\cite{fiduccia1972polynomial}. By performing the evaluation in this way, the space complexity of performing polynomial operations becomes roughly $p/k$. Using Bringmann's algorithm, we can choose $p\le \tilde O(t)$. The time complexity of one batch is $\tilde O(n+t)$, so the total time is $\tilde O((n+t)k)$.

\paragraph{The Approximation Algorithm.} All of our subset sum algorithms can be modified to work even if we want to check if there is a subset with sum in a specified \emph{range}, instead of being a fixed value. This property is used in our ``weak'' $\eps$-approximation algorithm (Theorem~\ref{thm:approx}). Our approximation algorithms also use ideas from previous subset sum approximation algorithms~\cite{MuchaW019, kellerer2003efficient} such as rounding elements and separating elements into bins of ``small'' and ``large'' elements.

\subsection{Other Related Work on Subset Sum}

The first $\tilde O(n+t)$-time (randomized) algorithm for Subset Sum was given by Bringmann~\cite{karl} (see also \cite{jw19} for log-factor improvements). 
For deterministic algorithms, the best known running time is $\tilde O(n+t\sqrt{n})$  \cite{Koiliaris}. 
Abboud, Bringmann, Hermelin, and  Shabtay~\cite{abboud2019seth} proved that there is no $t^{1-\delta}\cdot 2^{o(n)}$ time algorithm for \SS for any $\delta>0$, unless the Strong Exponential Time Hypothesis (SETH) fails.
Recently, Bringmann and Wellnitz~\cite{bringmann2020nearlineartime} showed an $\tilde O(n)$-time algorithm for  Subset Sum instances that satisfy certain density conditions.

Bansal, Garg, Nederlof, and Vyas \cite{bansal2018faster} obtained space-efficient algorithms for Subset Sum with running time exponential in $n$ but polynomial in $\log t$. They showed that Subset Sum (and Knapsack) can be solved in $2^{0.86 n}\cdot \poly(n,\log t)$ time and $\poly(n,\log t)$ space (assuming random read-only access to exponentially many random bits).
Note that currently the fastest algorithm in this regime uses $2^{n/2}\cdot \poly(n,\log t)$ time and $2^{n/4}\cdot \poly(n,\log t)$ space \cite{horowitz1974computing,schroeppel1981t} (the space complexity was recently slightly improved by Nederlof and W\k{e}grzycki~\cite{nederlof2020improving}).

Limaye, Mahajan, and Sreenivasaia~\cite{Limaye2012complexity} proved that Unary Subset Sum (where all numbers are written in unary) is $\mathsf{TC}^0$-complete. This implies that Subset Sum can be solved by constant-depth circuits with a \emph{pseudopolynomial} number of MAJORITY gates.

It is well-known that \SS also has a fully polynomial-time approximation scheme (FPTAS), which outputs a subset with sum in the interval $[(1 -\eps) t',t']$, where $t'$ is the maximum subset sum not exceeding $t$.
The most efficient known approximation algorithms are~\cite{kellerer2003efficient} which runs in time about $\tilde{O}(\min \{ n/\eps, n+1/\eps^2\})$ and space $\tilde{O}(n + 1/\eps)$, and ~\cite{GalJLMS14} which runs in time $O(n^2/\eps)$ and space $O((\log t)/\eps + \log n)$. Bringmann and Nakos \cite{Bringmann-approx} showed that the  time complexity of approximating Subset Sum cannot be improved to $(n+1/\eps)^{2-\delta}$ for any $\delta>0$, unless Min-Plus Convolution can be computed in truly subquadratic time.

\section{Preliminaries}
Let $[n]$ denote $\{1,2,\dots,n\}$, and let $\log x$ denote $\log_2 x$.
We will use $\tilde{O}(f)$ (and $\tilde \Omega, \tilde \Theta$) to denote a bound omitting $\polylog (f)$ factors.

	\subsection{Finite Fields}
	\label{sec:prelim}
We will use finite field arithmetic in our algorithm. Given $k\ge 1$ and prime $p$, we construct the finite field $\F_{q}$ of order $q=p^k$ by finding an irreducible polynomial of degree $k$ over $\F_p$, which can be solved by a (Las Vegas) randomized algorithm in time $\poly(\log p,k) \le  \polylog q$ (e.g., \cite{shoup1994fast}). 

For $k=2$, an irreducible polynomial can be found deterministically in $O(q)$ time and $O(\log q)$ space: it suffices to find a quadratic non-residue modulo $p$, since $x^2 - a$ can be factored iff $a$ has a square root modulo $p$. To find a quadratic non-residue, we simply enumerate all $a \in \F_p^*$ and check if $a$ has a square root in $O(p)$ time and $O(\log p)$ space. The total running time is $O(p^2)=O(q)$.

\subsection{Expanders}

To keep the randomness in our algorithms low, we use ``strongly explicit'' constructions of expander graphs.

\begin{definition}[Expander graphs]
	An $n$-vertex undirected graph $G$ is an $(n,d,\lambda)$-\emph{expander graph} if $G$ is $d$-regular and $\lambda(G)\le \lambda$, where $\lambda(G)$ is the second largest eigenvalue (in absolute value) of the normalized adjacency matrix of $G$ (i.e., the adjacency matrix of $G$ divided by $d$). 
	
	For constant $d \in \N$ and $\lambda <1$, a family of graphs $\{G_n\}_{n \in \N}$ is a $(\lambda,d)$-\emph{expander graph family} if for every $n$, $G_n$ is an $(n,d,\lambda)$-expander graph.
	\end{definition}
	
	\begin{theorem}[Explicit Expanders~\cite{margulis1973explicit,gabber1981explicit}]\label{theo:explicit-expander}
	There exists a $(\lambda,d)$-expander graph family $\{G_n\}$ for some constants $d\in \N$ and $\lambda<1$, along with an algorithm that on inputs $n\in \N,v\in [n],i\in [d]$ outputs the $i$-th neighbor of $v$ in graph $G_n$ in $\poly \log n$ time and $O(\log n)$ space.
	\end{theorem}

	We also need the following tail bound, showing that a random walk on an expander graph behaves similarly to a sequence of i.i.d. randomly chosen vertices.

	\begin{theorem}[Expander Chernoff Bound, \cite{gillman1998chernoff}]
	\label{theo:expander-chernoff}
	Let $G$ be an $(n,d,\lambda)$-expander graph on the vertex set $[n]$. Let $f: [n]\to \{0,1\}$, and $\mu = \Ex_{v\in [n]}f(v)$. Let $v_1,v_2,\dots,v_t$ be a random walk on $G$ (where $v_1$ is uniformly chosen, and $v_{i+1}$ is a random neighbor of $v_i$ for all $i$). Then for $\delta>0$, 
	$$\Pr_{v_1,\dots,v_t}\left [\frac{1}{t}\sum_{i=1}^t f(v_i)<\mu -\delta\right ]\le e^{-(1-\lambda)\delta^2t/4}.$$
	\end{theorem}

\subsection{Polynomial operations}

We will also utilize the classic multipoint evaluation algorithm for (univariate) polynomials.

\begin{theorem}[Fast multipoint evaluation \cite{fiduccia1972polynomial}]\footnote{A good modern reference is {\cite[Section 10.1]{von2013modern}}.}
\label{thm:multipoint} Every univariate polynomial over $\F$ of degree less than $n$ can be evaluated at $n$ points in $O(M(n)\log n)$ time and $ O(n\log n\cdot \log |\F|)$ space, where $M(n)$ is the time complexity of polynomial multiplication over $\F$.
\end{theorem}

Note that we can perform polynomial multiplication over $\F_q$ using FFT in $\tilde O(n\cdot \polylog(q))$ time.

\section{Kane's Number-Theoretic Approach}
\label{sec:kane}
In this section, we will review Kane's number-theoretic technique for solving \SS~\cite{Kane10}. Kane used this technique to obtain a \emph{deterministic} algorithm for Subset Sum, but it is straightforward to adapt it to the \emph{randomized} setting.
We will also improve Kane's deterministic algorithm.

As discussed in the introduction section, Kane utilizes the following simple fact.
\begin{lemma}
\label{lem:num-theo-primepower}
Let $q=p^k$ be a prime power and let $a$ be a positive integer. Let $\F_q^*$ denote the set of non-zero elements in the finite field $\F_q$. Then
\begin{itemize}
    \item If $q-1$ does not divide $a$, then  $\sum_{x\in \F_q^*} x^a=0$.
    \item If $q-1$ divides $a$, then 
    $\sum_{x\in \F_q^*} x^a=-1$.
\end{itemize}
\end{lemma}
\begin{proof}
The multiplicative group $\F_q^*$ is cyclic. Let $g$ be a generator of $\F_q^*$. For $q-1 \nmid a$, $g^a \neq 1$ and we have
$$\sum_{x\in \F_q^*}x^a = \sum_{i=0}^{q-2} g^{ai}= \frac{1-g^{a(q-1)}}{1-g^a} = 0.$$
When $q-1\mid a$, 
\begin{equation*}\sum_{x\in \F_q^*}x^a = (q-1)\cdot 1 = -1. 
\end{equation*}
The proof is complete.\hfill \end{proof}

\begin{cor}
\label{cor:kane}
Let $f(x) = \sum_{i=0}^d c_i x^i$ be a polynomial of degree at most $d$, where  coefficients $c_i$ are integers.
Let $\F_q$ be the finite field of order $q = p^k \ge d+2$.
For $0\le t \le d$, define
\[r_t := \sum_{x\in \F_q^*} x^{q-1-t}f(x) \in \F_q.\]
Then $r_t = 0$ if and only if $c_t$ is divisible by $p$.
\end{cor}
\begin{proof}
We have
$$r_t = \sum_{i=0}^d c_i \sum_{x\in \F_q^*}x^{i+q-1-t} = -c_t,$$
where the last equality follows from Lemma~\ref{lem:num-theo-primepower} and $|i-t| \le d < q-1$. Hence $r_t$ equals $0$ iff $-c_t$ is an integer multiple of the characteristic of $\F_q$.
\hfill
\end{proof}

The following lemma generalizes the main technique of Kane's algorithm for Subset Sum~\cite{Kane10}, extending it to randomized algorithms.

\begin{lemma}[Coefficient Test Lemma]
\label{lem:kane-improved}
Consider a polynomial $f(x)=\sum_{i=0}^d c_i x^i$ of degree at most $d$, with integer coefficients satisfying $|c_i|\le 2^w$.
Suppose there is a $T$-time $S$-space algorithm for evaluating $f(a)$ given $a\in \F_q$, where $q \le O(d+w)$.

Then, there is an $O((d+w)\cdot (T+w+\log d))$-time and $O(S+\log (dw))$-space algorithm using $\log w-\log \log w + O(1)$ random bits that, given $0\le t\le d$, tests whether $f(x)$ has a non-zero $x^t$ coefficient $c_t$, with at most $0.01$ one-sided probability of error.

Trying all random choices, this yields a deterministic algorithm in $O((d+w)\cdot  (T+w+\log d)\cdot  w/\log w)$ time and $O(S+\log (dw))$ space.
\end{lemma}

\begin{proof} We consider two cases, based on whether the degree of $f$ is ``high'' or ``low'' compared to the cofficient size.

{\bf Case 1: $d> w^2$.} By the Prime Number Theorem, we can pick $B=O(d)$  such that the interval $[\sqrt{d+2},\sqrt{B}]$ contains at least $m = 100 w/\log{\sqrt{d+2}}$ primes.
We can deterministically find such primes $p_1,\dots,p_{m}$ by simple trial division, in $O(\log d)$ space and $O(\sqrt{B} \cdot (\sqrt{B})^{1/2}) \leq O(d^{3/4})$ total time.
Let $q_i:=p_i^2$. Then the $q_i$'s are prime powers in the interval $[d+2,B]$.

Randomly pick $i\in [m]$, and let $q:=q_i,p:=p_i$. We compute $r_t = \sum_{x\in \F_q^*} x^{q-1-t}f(x)$ by running the evaluation algorithm (and the modular exponentiation algorithm for computing $x^{q-1-t}$) $(q-1)$ times, in $(q-1)\cdot (T+O(\log q)) \le O(d(T+\log d))$ total time.
By Corollary~\ref{cor:kane}, $r_t=0$ if and only if $p$ is a prime factor of $c_t$. As $|c_t|\le 2^w$, a non-zero $c_t$ has at most $w/\log\sqrt{d+2}$ prime factors from the interval $[\sqrt{d+2},\sqrt{B}]$, so the probability that $p$ is a prime factor is at most $(w/\log\sqrt{d+2})/m = 1/100$. 

{\bf Case 2: $d\le w^2$.}
 By the Prime Number Theorem, setting $B = \Theta(d+w)$, the interval $[\max\{d+2,w\},B]$ contains at least $m=100 w/\log w$ primes. As in the first case, we can deterministically find such primes $p_1,\dots,p_{m}$ by trial division, in $O(\log B)$ space and $O(B\cdot \sqrt{B}) \le  O((d+w)w)$ total time. 
 
We randomly pick $i\in [m]$ and let $p:=p_i$. 
We compute
$r_t = \sum_{x\in \F_p^*} x^{p-1-t}f(x)$
by running the evaluation algorithm $p-1$ times, in $(p-1)\cdot (T+O(\log p)) \le O((d+w)(T+\log (d+w)))$ total time.
By Corollary~\ref{cor:kane}, $r_t=0$ if and only if $p$ is a prime factor of $c_t$.
As $|c_t|\le 2^w$, a non-zero $c_t$ has at most $w/\log w$ prime factors from the interval $[\max\{d+2,w\},B]$, so the probability that $p$ is a prime factor is at most $1/100$. 

To make the above tests deterministic, we simply iterate over all $i\in [m]$. Note that we only used finite fields of order $q=p$ or $q=p^2$, the latter of which can be constructed deterministically in $O(q)$ time and $O(\log q)$ space (see Section~\ref{sec:prelim}).
\hfill
\end{proof}

\begin{cor}
\SS can be solved deterministically in $O(n^3 t \log t/\log n)$ time and $O(\log( n  t))$ space.
\end{cor}
\begin{proof}
Define $f(x)= \prod_{i=1}^n (1+x^{a_i})$ as the generating function of the Subset Sum instance. It has degree $d = \sum_{i=1}^n a_i\le nt$, and integer coefficients at most $2^n$. Given $x\in \F_q$, $f(x)$ can be evaluated in $T=O(n\log t)$ time, $O(\log (nt))$ space.
By Lemma~\ref{lem:kane-improved} we have an $O(n^3t\log t/\log n)$-time $O(\log (nt))$-space deterministic algorithm.
\hfill
\end{proof}

\begin{remark}
\label{rem:kane}
Our deterministic algorithm is slightly different from Kane's original algorithm (which is essentially the ``Case 2'' in the proof of Lemma~\ref{lem:kane-improved}). 
Kane's deterministic algorithm needs to find $\tilde \Omega(n)$ many primes $p>nt$ for constructing the finite fields. This could be time-consuming, as it takes square-root time to perform primality test\footnote{We could of course use more time-efficient algorithms such as AKS~\cite{aks04}, but it is not clear how to implement them in logspace, rather than poly-log space.}.
Our algorithm solves this issue by working over $\F_{p^2}$ instead of $\F_p$ (as in ``Case 1'')\footnote{This extension was already observed by Kane in {\cite[Section 3.1]{Kane10}}, but the purpose there was not to reduce the time complexity of generating primes.}.

Here we give an upper bound on the running time of Kane's original algorithm for finding primes. By the prime number theorem, in the interval $[nt, 100nt]$ there exist at least $nt/\log(nt) \ge n/\log n$ primes.
On the other hand, by \cite{baker2001difference}, the interval $[N,N+O(N^{0.525})]$ contains a prime, so we can find $\tilde \Omega(n)$ primes from the interval $[nt, nt + n\cdot (nt)^{0.525}]$. Hence, it suffices to search for all the primes from the interval $[nt, nt + \min\{99nt, n\cdot (nt)^{0.525}\}]$, in $O(\min\{nt, n\cdot (nt)^{0.525}\} \cdot \sqrt{nt})$ time and $O(\log(nt))$ space. So the total time complexity of Kane's algorithm is $\tilde O(\min\{nt, n\cdot (nt)^{0.525}\} \cdot \sqrt{nt} + n^3t) = \tilde  O(n^{2.025}t^{1.025}+n^3 t)$, where the first term becomes dominant when $t\gg n^{39}$.
\end{remark}

\section{Randomized Algorithm}
\label{sec:rand}

Our randomized algorithm uses \emph{color-coding} (i.e., random partitioning), which plays a central role in Bringmann's algorithm for  \SS~\cite{karl}. 
By randomly partitioning the input elements into groups, with high probability, the number of \emph{relevant} elements (i.e., elements appearing in a particular subset with sum equal to $t$) that any group receives is not much larger than the average. Once this property holds, then from each group we only need to compute the possible subset sums achievable by subsets of small size.

However, storing a uniform random partition of the elements will, in principle, require substantial space (at least $\Omega(n)$). We shall modify the random partitioning technique in a way that uses very low space. To do this, we need a way to quickly report the elements from a particular group, without explicitly storing the whole partition in memory. We formalize our requirements in the following definition. 

\begin{definition}[Efficiently Invertible Hash Functions]
\label{def:hash} For $1\le m \le n$, a family $\mathcal{H} = \{h\colon [n]\to [m]\}$ is a family of \emph{efficiently invertible hash functions with parameter $k(n) \geq 1$, seed length $s(n) \geq \log n$, and failure probability $\delta$}, if $|\mathcal{H}|=2^{s(n)}$ and the following properties hold:
\begin{itemize}
    \item[{\bf (Load Balancing)}] Let $S_i = \{x\in [n]: h(x)=i\}$.  For every set $S\subseteq [n]$ of size $m$, 
    $$\Pr_{h\sim \mathcal{H}}\Big [\max_{i\in[m]}|S \cap S_i| \le k(n)\Big] > 1-\delta.$$
    \item[{\bf (Efficient Invertibility)}] Given an $s(n)$-bit seed specifying $h \in \mathcal{H}$ and $i \in [m]$, all elements of $S_i$ can be reported in $O(|S_i|\cdot\poly\log n)$ time and $O(\log n)$ additional space.\footnote{Note that the output bits do not count towards the space complexity.}
\end{itemize}
\end{definition}


The goal of this section is to show that good implementations of efficiently invertible hash functions imply time and space efficient \SS algorithms. 

\begin{theorem}
\label{thm:algo-with-oracle}
Suppose for every constant $c\ge 1$, one can implement an efficiently invertible hash family with parameter $k(n)\ge 1$, seed length $s(n)\ge \log n$ and failure probability $1/n^{c}$.

Then, given random and read-only access to a string of $s(n)$ random bits, \SS can be solved with constant probability in $\tilde O(nt \cdot \poly(k(n)))$ time and $O(\log (nt))$ space. 
This implies that  \SS can be solved with constant probability in $\tilde O(nt \cdot \poly(k(n)))$ time and $O(s(n)+\log (t))$ space. 
\end{theorem}

Here we state two efficiently invertible hash families of different parameters and seed lengths. Their construction will be described in Section~\ref{sec:hash}.
\begin{theorem}
\label{thm:hash}
For any constants $c\ge 1$ and $\eps>0$,
\begin{enumerate}[(1)]
    \item there is an efficiently invertible hash family with parameter $k(n) = O(\log n)$, seed length $s(n) = O(\log n\log\log n)$ and failure probability $1/n^c$. And,
    \item there is an efficiently invertible hash family with parameter $k(n) = O(n^\eps)$, seed length $s(n) = O(\log n)$, and failure probability $1/n^c$. 
\end{enumerate}
\end{theorem}
Note that Theorem~\ref{thm:main} immediately follows from Theorem~\ref{thm:algo-with-oracle} and Theorem~\ref{thm:hash}.
\medskip

The algorithm given by Theorem~\ref{thm:algo-with-oracle} has a similar overall color-coding structure as in Bringmann's \SS algorithm~\cite{karl}, but it applies a generating function approach akin to Kane's \SS algorithm~\cite{Kane10}. In particular, on a \SS instance with target $t$, the algorithm constructs a polynomial $f(x)$ (depending on the input and the random bits), which always has a zero $x^t$ coefficient for NO-instances, and with $0.99$ probability has a non-zero $x^t$ term for YES-instances. Applying the Coefficient Test Lemma (Lemma~\ref{lem:kane-improved}), we can use the evaluation of $f(x)$ to solve \SS.

\subsection{Description of the algorithm}

We now proceed to proving Theorem~\ref{thm:algo-with-oracle}. We describe an algorithm $\mathtt{Evaluate}$ that evaluates such a polynomial $f(x)$ over $\F_q$. 
To keep the pseudocode clean, we do not specify the low-level implementation of every step. In the next section, we will elaborate on how to implement the pseudocode in small space.


\paragraph*{Layer Splitting.} We use the layer splitting technique from Bringmann's \SS algorithm~\cite{karl}. Given the input (multi)set of integers $A=\{a_1,\dots,a_n\}$ and target $t$, we define $L_i:= A\cap (t/2^i, t/2^{i-1}]$ for $i=1,2,\dots,\lceil \log n\rceil - 1$, and let $L_{\lceil \log n \rceil} = A \backslash (L_1\cup \dots \cup L_{\lceil \log n \rceil-1})$. 

In the following analysis, for a given YES-instance, fix an arbitrary solution set $R\subseteq A$ which sums to $t$.
The elements appearing in $R$ are called \textit{relevant}, and the elements in $A\setminus R$ are irrelevant.
Observe that the $i$-th layer $L_i$ can only contain at most $2^i$ relevant elements. 


\paragraph*{Two-level random partitioning.}

For each layer $L_i$, we apply two levels of random partitioning, also as in Bringmann's algorithm.
In the first level, we use our efficiently invertible hash function (specified by a random seed $r_1$ of length $s(n)$) to divide $L_i$ into $\ell=2^i$ groups $S_1,\dots,S_\ell$, so that with high probability each $S_j$ contains at most $k$ relevant elements.

In the second level, we use a pairwise independent hash function (specified by random seed $v$ of length $s'=O(\log n)$) to divide each group $S_j$ into $k^2$ mini-groups $T_1,\dots,T_{k^2}$, so that with $\ge 1/2$ probability each relevant element in $S_j$ appears in a \emph{distinct} $T_i$, isolated from the other relevant elements in $S_j$.

We repeat the second-level partitioning for $m=O(\log n)$ rounds, increasing the success probability to $1-1/\poly(n)$. However, instead of using fresh random bits each round (in which the total seed length would become $ms' = O(\log^2 n)$), we use the standard expander-walk sampling technique:  we pseudorandomly generate the seeds $v_1,\dots,v_m$ by taking a length-$m$ random walk on a strongly explicit expander graph (Theorem~\ref{theo:explicit-expander}), where each vertex $v_i$ represents a possible $s'$-bit seed. The random walk is specified by the seed $r_2$ which only has $s'+O(m) = O(\log n)$ bits. 

\begin{figure}[H]
\begin{framed}
\begin{algorithmic}[1]
\REQUIRE $\mathtt{Evaluate}(x, A, t, r_1, r_2)$
\STATE Let $L_i:= A\cap (t/2^i, t/2^{i-1}]$ for $i=1,2,\dots,\lceil \log n\rceil - 1$
\STATE Let $L_{\lceil \log n \rceil} := A \setminus (L_1\cup \dots \cup L_{\lceil \log n \rceil-1})$
\STATE $u = 1$
\FOR{$i=1,\ldots, \lceil \log n\rceil$}
	\STATE $u = u\cdot \mathtt{PartitionLevel1}(x, L_i,i, r_1,r_2)$
\ENDFOR
\STATE \algorithmicreturn\ $u$
\end{algorithmic}
\end{framed}
\end{figure}
\begin{figure}[H]
\begin{framed}
\begin{algorithmic}[1]
\REQUIRE $\mathtt{PartitionLevel1}(x, L, i, r_1,r_2 )$ 
\STATE Let $\ell = 2^i$
\STATE Partition
$L = S_1 \ \dot{\cup}   \cdots \dot{\cup} \  S_{\ell}$
using the efficiently invertible hash function $h$ with parameter $k(n)$, specified by seed $r_1$
\STATE $u = 1$
\FOR{$j=1,\ldots,\ell$}
	\STATE $u=u\cdot \mathtt{PartitionLevel2}(x, S_j, k,r_1, r_2)$
\ENDFOR
\STATE \algorithmicreturn\ $u$
\end{algorithmic}
\end{framed}
\end{figure}
\begin{figure}[H]
\begin{framed}
\begin{algorithmic}[1]
\REQUIRE $\mathtt{PartitionLevel2}(x, S, k,r_1,r_2)$
\STATE $u=0, m = c' \log n$ (where $c'$ is a sufficiently large constant)
\STATE Let $v_1,\dots,v_m \in \{0,1\}^{O(\log n)}$ be the seeds extracted from $r_2$ using expander walk sampling.
\FOR{$j=1,\ldots,m $}
    \STATE Partition $S = T_1 \ \dot{\cup}   \cdots \dot{\cup} \  T_{k^2}$  using the pairwise-independent  hash function specified by seed $v_j$ 
    \STATE $w = \prod_{i=1}^{k^2} \big (1+\sum_{a\in T_i}x^a\big )$
	\STATE $u = u+w$
\ENDFOR
\STATE \algorithmicreturn\ $u$
\end{algorithmic}
\end{framed}
\end{figure}

\paragraph{Analysis.} Now we turn to analyzing the algorithm. It is easy to see that the output values of all three algorithms $\mathtt{Evaluate}$, $\mathtt{PartitionLevel1}$, and $\mathtt{PartitionLevel2}$ 
can be expressed as \emph{polynomials} in the variable $x$, with coefficients determined by the parameters (except $x$) passed to the function calls.
In the following we will state and prove the key properties satisfied by the polynomials output by these algorithms. 

In our analysis, the ``$x^\sigma$ term'' of a polynomial will be construed as the integer coefficient of $x^{\sigma}$ over $\Z$, although in our actual implementation we will perform all arithmetic operations (additions and multiplications) in $\F_q$ (as in the Coefficient Test of Lemma~\ref{lem:kane-improved}).

\begin{lemma}
\label{lem:colorcoding}
The output of $\mathtt{PartitionLevel2}(x,S,k,r_1,r_2)$ is a polynomial $P(x)$ of degree at most $k^2\cdot \max_{a\in S}a$, with non-negative coefficients bounded by $P(1)\le 2^{O(k^2 \log n)}$.  

Suppose $S$ contains at most $k$ relevant elements $R_S\subseteq S$, and let $\sigma = \sum_{a\in R_S}a$. Then $P(x)$ contains a positive $x^{\sigma}$ term with $1-\frac{1}{n^{c_2}}$ probability for any desired constant $c_2 \geq 1$.
\end{lemma}

\begin{proof} By construction, the output polynomial $P(x)$ has degree at most $k^2\cdot \max_{a\in S}a$ and non-negative coefficients bounded by $P(1)\le m\cdot (n+1)^{k^2} \le 2^{O(k^2 \log n)}$.

Let $h'_v\colon [|S|]\to [k^2]$ be the pairwise independent hash function specified by the seed $v$ of length $s' \leq O(\log |S|+\log k^2) \leq O(\log n)$ bits (such a pairwise independent family can be easily evaluated in logspace and essentially linear time, cf.\ \cite{vadhan2012pseudorandomness}).
For randomly chosen $v$, each pair of distinct relevant items are hashed into the same mini-group $T_i$ with probability $1/k^2$. By a union bound over $k(k-1)/2$ pairs of relevant items, the probability that all relevant items are isolated (end up in distinct mini-groups) is at least $1-\frac{k(k-1)/2}{k^2}\ge 1/2$.

The seeds $v_1,\dots,v_m$ are generated by taking a length-$m$ random walk on a strongly explicit $(2^{s'}, d, \gamma)$-expander graph for some constants $d\in \N, \lambda <1$ (Theorem~\ref{theo:explicit-expander}). By the Expander Chernoff Bound (Theorem~\ref{theo:expander-chernoff}), over the choice of random seeds $v_1,\dots,v_m$,
\begin{align*}
	 & \Pr\left [\frac{1}{m}\sum_{j=1}^m [\text{$v_j$ isolates all relevant items}]<1/2 -1/4\right ]
	 \le   e^{-(1-\lambda)(1/4)^2m/4}.  \end{align*}
Choosing $m := c'\log n$ for a sufficiently large constant $c'$, the above probability is at most $n^{-c_2}$ for any desired constant $c_2 \geq 1$.  
Hence with probability at least $1-n^{-c_2}$, there is a $j\in [m]$ such that the seed $v_j$ isolates all relevant items into different mini-groups $T_1,\dots,T_{k^2}$.
In such a case, the product \[\prod_{i=1}^{k^2} \left(1+\sum_{a\in T_i}x^a\right)\] must have a positive $x^\sigma$ term. Hence, the sum $u$ also has a positive $x^\sigma$ term.  
\end{proof}

\begin{lemma}
\label{lem:colorcodinglayer}
The output of $\mathtt{PartitionLevel1}(x,L,i,r_1,r_2)$ is a polynomial $Q(x)$ of degree at most $2k^2 t$, with non-negative coefficients bounded by $Q(1) \le 2^{O(2^i k^2 \log n)}$.

Let $R_L\subseteq L$ be the subset of relevant elements, and $\sigma = \sum_{a\in R_L}a$.
Then $Q(x)$ contains a positive  $x^\sigma$ term with $1-\frac{1}{n^{c_3}}$ probability for any desired constant $c_3 \geq 1$.
\end{lemma}

\begin{proof}

By Lemma~\ref{lem:colorcoding}, the degree of every return value of $\mathtt{PartitionLevel2}(x,S_j,k,r_1,r_2)$ as a polynomial is at most $k^2\cdot \max_{a\in S_j}a$.  Hence, the degree of $Q(x)$ as a polynomial is at most $\ell \cdot k^2\cdot \max_{a\in L_i}a \le \ell \cdot k^2\cdot  t/2^{i-1} = 2k^2 t$. And $Q(1)$ is the product of $\ell=2^i$ many $P(1)$'s, which is at most $2^{O(2^i k^2 \log n)}$.

Recall that $|R_L| \le \ell = 2^i$. 
Let $R_j:= R_L \cap S_j$ and $\sigma_j := \sum_{a\in R_j}a$.
By the properties of efficiently invertible hash families, we have $|R_j| \le k$ for every $j\in [\ell]$ with probability $1-n^{-c}$.

Thus by Lemma~\ref{lem:colorcoding} and the union bound, with probability $1-n^{-c_3}$ (for any desired constant $c_3\ge 1$), for every $j\in [\ell]$, the output polynomial $P(x)$ of $\mathtt{PartitionLevel2}(x, S_j, k,r_1, r_2)$ has a positive $x^{\sigma_j}$ term. 
In such a case, $Q(x)$ has a positive $x^{\sigma_1+\dots+\sigma_\ell}=x^\sigma$ term.
\end{proof}

\begin{lemma}
\label{lem:evaluate}
The output of $\mathtt{Evaluate}(x,A,r_1,r_2)$ is a polynomial $S(x)$ of degree at most $2k^2 t\cdot \lceil \log n\rceil $, with non-negative coefficients bounded by $S(1) \le 2^{O(\min\{n,t\}\cdot k^2 \log n)}$.

Let $R\subseteq A$ be the subset of relevant elements, and $t = \sum_{a\in R}a$. Then $S(x)$ contains a positive $x^t$ term with $1-\frac{1}{n^{c_4}}$ probability, for any desired constant $c_4 \geq 1$.
\end{lemma}

\begin{proof}
By Lemma~\ref{lem:colorcodinglayer}, each call to $\mathtt{PartitionLevel1}(x,L_i,i,r_1,r_2)$ returns a polynomial $Q(x)$ of degree at most $2k^2 t$ with non-negative coefficients, and $Q(1)\le 2^{O(2^i k^2 \log n)}$. Hence, the degree of the product polynomial $S(x)$ is at most $2k^2 t\cdot \lceil \log n\rceil$. And,
\[S(1)\le  2^{O(k^2 \log n) \cdot (2^1+2^2+\dots +2^{i'})} \le 2^{O(2^{i'} k^2 \log n) },\] where $i'\le \lceil \log n\rceil$ and $t/2^{i'-1}\ge 1$ (for non-empty $L_1,\dots,L_{i'}$).

Let $R_i:= R \cap L_i$ and $\sigma_i = \sum_{a\in R_i}a$.
By Lemma~\ref{lem:colorcodinglayer} and the union bound, with probability at least $1-n^{-c_4}$ (for any desired constant $c_4 \geq 1$), for every $i=1,\dots,\lceil \log n\rceil$, the output polynomial $Q(x)$ of 
$\mathtt{PartitionLevel1}(x, L_i,i, r_1,r_2)$ contains a positive  $x^{\sigma_i}$ term. In such a case, the final product $S(x)$ has a positive $x^{\sigma_1+\dots+\sigma_{\lceil \log n\rceil}}=x^t$ term.
\end{proof}

\subsection{Implementation}
 Now we describe in more detail how to implement $\mathtt{Evaluate}$, $\mathtt{PartitionLevel1}$, and $\mathtt{PartitionLevel2}$ in logarithmic space (assuming read-only access to random seeds $r_1,r_2$), and analyze the time complexity.
 
 \begin{lemma}
 \label{lem:implement}
Assuming read-only access to random seeds $r_1,r_2$, the procedure $\mathtt{Evaluate}$ (where arithmetic operations are over $\F_q$ with $q=\Omega(t)$) runs in $ \tilde O(n\cdot  \poly (k,\log q))  $ time and $O(\log(nq))$ working space.
 \end{lemma}
 
 \begin{proof}
Recall $A = \{a_1,\ldots,a_n\}$ is the set of input integers. In $\mathtt{Evaluate}(x,A,t,r_1,r_2)$, we do not have enough space to collect all elements of $L_i$ and pass them to $\mathtt{PartitionLevel1}$. 
Instead, in $\mathtt{PartitionLevel1}$, we partition the set $A$ into $\ell$ groups $S'_1,\dots,S'_\ell$ using an efficiently invertible hash function $h\colon [n] \to [\ell]$, specified by seed $r_1$
(on different layers, the value of  $\ell=2^i$ is different, but we will use the same seed $r_1$ to generate $h\colon [n]\to [\ell]$). Next, we define the partition of $L_i = S_1 \  \dot{\cup} \cdots \dot{\cup} \ S_\ell$ by $S_j:=S'_j \cap L$.
By the efficient invertibility property, assuming read-only access to the seed $r_1$, we can iterate over the elements of $S_j$ in $ O(|S_j'|\cdot \polylog(n))$ time and $O(\log n)$ space, by iterating over $|S_j'|$ and ignoring the elements that do not belong to $L_i$. Hence, iterating over $S_1,\dots,S_\ell$ takes $\tilde O(n)$ time in total.

In $\mathtt{PartitionLevel2}(x,S,k,r_1,r_2)$, we iterate over $j=1,\ldots,m$, and compute each $v_j$ in $O(\log n)$ space by following the expander walk. In the loop body, for every $T_i$, we compute $\left(1+\sum_{a\in T_i}x^a\right)$ in $O(\log n+\log q)$ space by iterating over $S$ and ignoring the elements that do not belong to $T_i$.
In total, we iterate over $S$ for $m k^2$ passes during the execution of $\mathtt{PartitionLevel2}(x,S,k,r_1,r_2)$.
In $\mathtt{PartitionLevel1}$, the procedure $\mathtt{PartitionLevel2}$ is called once for each $S_j$, so the total time is $\tilde O(n \cdot  mk^2\cdot \polylog(q))$.
Therefore, over all $\lceil \log n\rceil$ layers, the total time complexity of $\mathtt{Evaluate}$ is $\tilde O(n\cdot \poly(k\log q))$, and the space complexity is $O(\log(nq))$.
\end{proof}

We complete this section by proving Theorem~\ref{thm:algo-with-oracle}, which demonstrates how to solve \SS efficiently from an efficiently invertible hash family by computing $\mathtt{Evaluate}$ on many values of $x$.

\begin{proofof}{Theorem~\ref{thm:algo-with-oracle}}
By Lemma~\ref{lem:evaluate}, the output of $\mathtt{Evaluate}$ is a polynomial of degree $d=O(t\cdot \poly(k,\log n))$, with non-negative coefficients bounded by $2^{w}$ where $w=O(\min\{n,t\}\cdot \poly(k,\log n) )$.

Applying the Coefficient Test Lemma (Lemma~\ref{lem:kane-improved}) to the efficient $\mathtt{Evaluate}$ algorithm of Lemma~\ref{lem:implement}  with running time $T=\tilde O(n\cdot \poly(k,\log q))$ where $q\le O(d+w)$, we can test whether the output of $\mathtt{Evaluate}$ as a polynomial contains a positive $x^t$ term in $\tilde O((d+w)(T+w)) \le \tilde O(nt \cdot \poly(k))$ time and $O(\log (nt))$ space.
\end{proofof}

To complete our randomized algorithm for \SS, it remains to provide the efficiently invertible hash families claimed in Theorem~\ref{thm:hash}, which we do next.

\section{Construction of Efficiently Invertible Hash Families}
\label{sec:hash}

Now we show how to construct the efficiently invertible hash families (Definition~\ref{def:hash}) used by the function $\mathtt{PartitionLevel1}$ in our algorithm from Section~\ref{sec:rand}.
Our construction has a similar structure to the hash family constructed by Celis, Reingold, Segev, and Wieder \cite{balls}, which achieves the load-balancing property of our hash functions. By making several modifications to their construction, the hash family can be made efficiently invertible.


\begin{theorem}
\label{thm:bijloglog}
For any $c>0$, and integers $1\le m\le n$ which are both powers of 2, there is a family $\caH_{n,m}$ of \emph{bijections} $h\colon  [n] \to [m] \times  [n/m]$ satisfying the following conditions.
\begin{itemize}
    \item Each function $h\in \mathcal{H}_{n,m}$ can be described by an $O(\log n\log \log n)$-bit seed.
    \item Given the seed description of $h\in \caH_{n,m}$, $h(x)$ can be computed in $O((\log \log n)^2)$ time and $O(\log n)$ space for any $x\in [n]$, and $h^{-1}(i,j)$ can be computed in $O((\log \log n)^2)$ time and $O(\log n)$ space for any $i\in [m], j\in [n/m]$.
    \item For $i\in [m]$, let  $S_i: = \{x \in [n]: h(x) = (i,j)\text{ for some $j\in [n/m]$}\}$. There is a constant $\gamma>0$ such that for every set $S\subseteq [n]$ of size $m$, 
    $$\Pr_{h\in \caH_{n,m}}\Big [\max_{i\in[m]}|S \cap S_i| \le \gamma \log n\Big] > 1-\frac{1}{n^c}.$$
\end{itemize}
\end{theorem}

\begin{remark}
The original construction of \cite{balls} achieves the optimal load-balancing parameter $\frac{\gamma \log n}{\log \log n}$, rather than $\gamma \log n$. Such a bound is also achievable for us; for simplicity we state a weaker version here, which is sufficient for our purpose.
\end{remark}

We observe that Theorem~\ref{thm:bijloglog}
immediately implies an efficiently invertible hash family with the desired parameters claimed in Theorem~\ref{thm:hash}(1). In particular, to invert a hash value $i$, we simply iterate over $j\in [n/m]$ and output the unique $x\in [n]$ such that $h(x)=(i,j)$.

\begin{cor}
\label{cor:hashloglog}
For any constant $c\ge 1$, there is an efficiently invertible hash family with parameter $k(n) = O(\log n)$, seed length $s(n) = O(\log n\log\log n)$ and failure probability $1/n^c$.
\end{cor}

We now describe the construction of $\caH_{n,m}$. The analysis of correctness is basically the same as~\cite{balls}; for completeness, we include this analysis in Appendix~\ref{apx:balance}.

We begin with a construction of an almost $k$-wise independent hash family.

\begin{definition}
A family $\mathcal{F}$ of functions $f\colon [u]\to [v]$ is $k$-wise $\delta$-dependent if for any  $k'\le k$ distinct elements $x_1,\dots,x_{k'}\in [u]$, the statistical distance between the distribution $(f(x_1),\dots,f(x_{k'}))$ where $f$ is uniformly randomly chosen from $\mathcal{F}$ and the uniform distribution over $[v]^{k'}$ is at most $\delta$.
\end{definition}

\begin{lemma}[\cite{AlonGHP90, meka2014fast}]
\label{lem:kwise}
Let $k\cdot \ell \leq O(\log n)$, $w\leq O(\log n)$, and $\delta = 1/\poly(n)$. There is a $k$-wise $\delta$-dependent family $\mathcal{H}$ of functions from $\{0, 1\}^w$ to $\{0, 1\}^\ell$, where each $h \in \mathcal{H}$ can be specified by an $O(\log n)$-bit seed, and each $h$ can be evaluated in $\poly\log n$ time.
 \end{lemma} 

Our construction of the family of bijections in Theorem~\ref{thm:bijloglog} has a $d$-level structure, where $d=O(\log \log n)$. First we assign some parameters:
\begin{itemize}
    \item $m_0=m, m_i = m_{i-1}/2^{\ell_i}$ for every $i\in [d]$;
    \item  $\ell_i = \left \lfloor (\log m_{i-1})/4 \right \rfloor$ for $i \in [d-1]$, and $\ell_d = \log m - \sum_{i=1}^{d-1} \ell_i$;
    \item $k_i \ell_i = \Theta(\log n)$, and $k_i$ is even for every $i\in [d-1]$;
    \item $k_d = \Theta(\log n/\log \log n)$;
    \item $\delta = 1/\poly(n)$.
\end{itemize} The constant factors $d,k_i,\log (1/\delta)$ all depend on the constant $c$, and are specified further in the analysis of Appendix~\ref{apx:balance}.
 
For every $i\in [d]$, let \[n_i := m_i \cdot (n/m).\] For each $i \in [d]$, independently sample a $g_i\colon  [n_i]\to [2^{l_i}]$ from a $k_i$-wise $\delta$-dependent family using a $O(\log n)$-bit random seed (Lemma~\ref{lem:kwise}). The total seed length is thus $d\cdot O(\log n) \leq O(\log n \log \log n)$. For $i\in [d]$, we define a bijection $f_i\colon [2^{l_i}] \times [n_i] \to [n_{i-1}]$ by \[f_i(b,u) := (b \oplus g_i(u))\circ u,\] where $\circ$ stands for concatenation.
Note that $f_i^{-1}(\cdot)$ can be computed with one evaluation of $g_i(\cdot)$. 

Now we define the bijections $h\colon  [n] \to [m] \times [n/m]$. Given input $x_0\in [n]$, let \[(b_i,x_i) =f_i^{-1}( x_{i-1})\] for $i = 1,2,\dots,d$.
Then we define \[h(x_0) := (b_1\circ \dots \circ b_d, x_d).\]  
To compute the inverse $x_0=h^{-1}(b_1\circ \dots \circ b_d, x_d)$, we can simply compute $x_{i-1}=f_i(b_i,x_i)$ for each $i=d,d-1,\dots,1$, and eventually find $x_0$. This also shows that $h$ is a bijection.

\begin{remark}
\label{rem:tree}
The above construction of the bijection $h\colon [n] \to [m]\times [n/m]$ can be naturally viewed as a depth-$d$ tree structure, similar to~\cite{balls}. 
The root node (at level 0) represents an array containing all elements of $[n]$ in ascending order.
For $0\le i\le d$, the $i$-th level of the tree has $2^{\ell_1+\ell_2+\dots + \ell_i} = n/n_i$ nodes, each representing a length-$n_i$ array.
A node $B$ at level $(i-1)$ has $2^{\ell_i}$ children $B_1,\dots,B_{2^{\ell_i}}$, which form a partition of the elements in array $B$:
the $u$-th element of array $B_{b}$ equals the $f_i(b,u)$-th element of array $B$.
At level $d$, there are $m$ leaf nodes $S_1, \dots S_m$, forming a partition of $[n]$, where the $u$-th element of array $S_b$ is $h^{-1}(b,u)$.
\end{remark}

By reducing the depth $d$ to a constant in the construction, we can achieve $O(\log n)$ seed length, but with a worse load-balancing parameter $k(n) = n^{\eps}$, for any $\eps>0$.

\begin{cor}
\label{cor:hasheps}
For any constants $c\ge 1$ and $\eps>0$, there is an efficiently invertible hash family with parameter $k(n) = O(n^\eps)$, seed length $s(n) = O(\log n)$, and failure probability $1/n^c$.
\end{cor}

The proof sketch is deferred to Appendix~\ref{apx:balance}.

\section{Time-Space Tradeoffs for Subset Sum}
In this section we present an algorithm achieving a time-space tradeoff for Subset Sum. Our algorithm uses several algebraic and number-theoretic ideas in order to trade more space for a faster running time.

\begin{reminder}{Theorem~\ref{thm:main_tradeoff}}
For any parameter $1\le k\le \min\{n,t\}$, there is a randomized algorithm for Subset Sum with $0.01$ one-sided error probability, running in $\tilde O((n+t)\cdot k)$ time and $O((t/k)\polylog (nt))$ space.
\end{reminder}

Our algorithm also uses Bringmann's framework combined with Kane's number-theoretic technique, as described in Section~\ref{sec:rand}. Recall that in the algorithm of Section~\ref{sec:rand}, assuming we use the hash family from Theorem~\ref{thm:hash}(1) with parameter $O(\log n)$,  we evaluated a particular generating function (a polynomial of degree at most $t\polylog n$, by Lemma~\ref{lem:evaluate}) at all points $b \in \F_q^{*}$, where $q = t \polylog n$ was a randomly chosen prime power (as described in Lemma~\ref{lem:kane-improved}).

Here, our new idea is to perform the evaluation in $S$ batches, where each batch has $(q-1)/S$ points to evaluate. 
In one batch, letting the evaluation points be $b_1,b_2,\dots, b_{(q-1)/S}$, we define a polynomial \[B(x) := (x-b_1)(x-b_2)\cdots(x-b_{(q-1)/S})\] of degree $(q-1)/S$.
Then we run the algorithm of Section~\ref{sec:rand}, with the following key modifications (the first were used in Bringmann's $\tilde O(n+t)$-time $\tilde O(t)$-space algorithm \cite{karl}, while the second one is new and crucial to the space improvement).
\begin{enumerate}[(1)]
    \item Instead of plugging in a specific value for $x$, we treat $x$ as a formal variable, and the intermediate results during computation are all expanded as a polynomial in $x$. 
    We use FFT for polynomial multiplication.
    \item All polynomials are computed modulo $B(x)$. Since $B(x)$ has degree $(q-1)/S$, the polynomial operations now only take $O((q/S)\log q)$ space.
\end{enumerate}
Finally we obtain the generating function modulo $B(x)$, which agrees with the original polynomial on the evaluation points $b_i$. We can evaluate at these points using Theorem~\ref{thm:multipoint} in $(q/S)\polylog q \le (t/S)\polylog(nt)$ time and space (recall that $q=t \polylog n$).

In order to run the above algorithm efficiently in low space, we have to make some adjustments to the lower-level implementation, which we elaborate in the following. 
\begin{itemize}
    \item 
In $\mathtt{PartitionLevel1}$, we need to compute the product of $\ell=2^i$ many polynomials of degree $\min \{d, (q-1)/S\}$ modulo $B(x)$, where $d=O((\log^2 n)\cdot 2t/\ell)$ (see Lemma~\ref{lem:colorcoding}). When $d<(q-1)/S$ is small, it might be slow to multiply them one by one. Instead, we divide them into groups, each having $\Theta(q/(Sd))$ polynomials. 
The polynomials in one group have total degree $O(q/S)$, and their product can be computed in $(q/S)\polylog q$ time and space, by multiplying them in a natural binary-tree structure.
There are $O\big (\ell\big /(q/(Sd))\big ) \le S\polylog n$ groups. We multiply the product of each group one by one modulo $B(x)$, in total time \[(S\polylog n)\cdot ((q/S)\polylog q)\le q  \polylog(nq),\] and $(q/S)\polylog q$ space. 

When $d\ge (q-1)/S$, we can simply multiply them one by one, in $q \polylog(nq)$ total time and $(q/S)\polylog q$ space.

\item In $\mathtt{PartitionLevel2}$, we need to compute the polynomial $ (1+\sum_{a\in T_i}x^a)$ modulo $B(x)$. When $a\gg (q-1)/S$, this is not easy to compute efficiently in $(q/S)\polylog q$ space. To resolve this issue, we will carefully pick the evaluation points $b_1,\dots,b_{(q-1)/S}$ so that $B(x)$ \emph{only has two terms}, i.e., it has the form $B(x) = x^{(q-1)/S}-h$. Then, $x^a \bmod B(x)$ is always a monomial, which can be easily computed in $\polylog(aq)$ time. 
\end{itemize}

Hence, one batch of evaluation can be performed in $\tilde O(n+t)$ time and $O((t/S)\polylog(nt))$ space. The total time complexity is $\tilde O((n+t)S)$.

Now we show how to pick the evaluation points $b_1,\dots,b_{(q-1)/S}$ so that $B(x)$ always has only two terms. We assume $S$ is a divisor of $q-1$. We use the following algebraic lemma:

\begin{lemma}\label{lemma:cool}
Let $S$ divide $q-1$. The set $\F_q^{*}$ can be partitioned disjointly into $S$ sets $P_0,\ldots,P_{S-1}$ such that for all $j=0,\ldots,S-1$ there is a polynomial $B_j(x)$ of two terms that vanishes only on $P_j$.
\end{lemma}

\begin{proof} Take a generator $g$ of $\F_q^*$.\footnote{To check whether some $h \in \F_q^*$ is a generator, we can simply enumerate all factors $r$ of $q-1$ and check if $h^{r} = 1$. This takes at most $O(\sqrt{q}\polylog(q))$ time. There are $\phi(q-1) \ge  \Omega(q^{0.99})$ many generators in $\F_q^{*}$ (see e.g., {\cite[Theorem 327]{hardy75}}), so we can find a generator with a (Las Vegas) randomized algorithm in $o(q)$ time.} 
For $j=0,\dots,S-1$, define the set of points
$$P_j := \{ g^{aS+j} \mid  a=0,\dots ,(q-1)/S-1\},$$
and the polynomial
$$B_j(x) := x^{(q-1)/S} - g^{j(q-1)/S}.$$
Note that $|P_j| = (q-1)/S$ and $P_0 \, \dot{\cup}\, P_1\, \dot{\cup} \cdots \dot{\cup}\, P_{S-1}$ is a partition of $\F_q^*$.
Now we want to show that, for all $0\le j \le S-1$, $B_j(x) = \prod_{b \in P_j} (x-b)$.

For every $b=g^{aS}\in P_0$,  $B_0(b) = g^{a(q-1)} - 1 = 0$. 
Since $|P_0| = (q-1)/S$, these are all the roots of $B_0(x)$, and we have $B_0 = \prod_{b\in P_0}(x-b)$. So the claim holds for $j=0$.

For $j\neq 0$, note that 
\begin{align*}
    \prod_{b\in P_j}(x-b) &=
    \prod_{a} (x-g^{aS+j})\\
    & =g^{j(q-1)/S} \prod_a(x/g^j - g^{aS})\\
    & = g^{j(q-1)/S}B_0(x/g^j)\\
    & = g^{j(q-1)/S}((x/g^j)^{(q-1)/S} - 1)\\
    & =B_j(x). \qedhere
\end{align*}
\hfill
\end{proof}

Hence, when $S$ divides $q-1$, we can simply use $P_0,P_1,\dots,P_{S-1}$ as the batches.

The following lemma ensures that, for any given parameter $1\le k\le \min\{n,t\}$, we can find a prime power $q = \tilde \Theta(t)$ such that $q-1$ has a divisor $S = \tilde \Theta(k)$.
\begin{lemma}
\label{lem:bv}
For all sufficiently large $R$ and  $4\le K\le R/4$, there are at least  $\Omega(R/\log^2 R)$ prime powers $q \in (R/2,R]$ such that $q-1$ has an integer divisor in the interval $[K/2,2K\cdot \log^{15}K]$.
\end{lemma}

We provide a proof of this lemma in Appendix~\ref{sec:bv}. It relies on the Bombieri-Vinogradov Theorem \cite{bombieri_1965, vin} from analytic number theory.

To finish the proof of Theorem~\ref{thm:main_tradeoff}, we apply Corollary~\ref{cor:kane} in a similar way as we did in proving the Coefficient Test Lemma. 
By Lemma~\ref{lem:evaluate}, there is some $w = \tilde \Theta(t \cdot \polylog n)$ such that the generating function has degree at most $w$ and integer coefficients of magnitude at most $2^w$.
For a given parameter $1\le k \le \min \{n,t\}$, we apply the Lemma~\ref{lem:bv} with some $R = \tilde \Theta(w)$ so that there are at least $100w$ prime powers $q$ in the interval $[w+2,R]$, such that $q-1$ has a divisor $S\in [K/2, \tilde O(K)]$, 
where $K=\tilde \Theta (k\polylog(nt) )$.
To find these prime powers $q$, we iterate over the interval $[w+2,R]$ and use AKS primality test, and then iterate over the interval $[K/2,\tilde O(K)]$ to find a divisor $S$ of $q-1$, using $\tilde O(RK) \le kt\polylog(nt)$ time and $\polylog(nt)$ space.
Then, we pick a random $q$ from them, and run the evaluation algorithm described above in $\tilde O ((n+t)S)$ total time and $(t/S)\polylog(nt)$ space. 

 




\section{Deterministic Algorithm}

In this section, we present a faster low-space deterministic algorithm for \SS. 

\begin{reminder}{Theorem \ref{thm:main_det}}
\SS can be solved deterministically in $\tilde O(n^2 t)$ time and $O(\log t\cdot \log^3 n)$ space.
\end{reminder}

Consider a subset sum instance $A = (a_1, a_2, \ldots, a_{n})$ with target sum $t$. In Section~\ref{sec:rand}, we modified Kane's algorithm so that it works modulo a prime $p \leq \tilde{O}(t)$, rather than needing $p \geq \Omega(nt)$. We achieved that by first randomly splitting $A$ into various sets $L_i$ (as Bringmann does), whereby in each $L_i$ we could estimate a nice upper bound on the number of elements that may contribute towards a subset sum of value $t$. To ensure that we only had the contribution of these sums and no larger sums, we used color-coding (and efficient hash functions to keep the space and the randomness low). As mentioned in the introduction, this approach gives fast low-space randomized algorithms, but it seems very difficult to derandomize color-coding quickly.

Here, to obtain a good deterministic algorithm, we give an alternative approach. As before, we split $A$ into $L_i$ lists. From there, we try to deterministically approximate the \emph{number of elements}, by only keeping track of the approximate logarithm of the number of elements, in a similar spirit to Morris's algorithm~\cite{morris1978counting} for small-space approximate counting. We start by defining a special polynomial product operation that will help us approximately count.


\begin{definition}[Product for Approximate Counting]\label{def:aprod}
	Let $\epsilon \in (0,1)$. Let $q_1(y) = 1+\sum_{i=1}^{d_1}v_iy^i$ and $q_2(y)= 1+\sum_{j=1}^{d_2}w_jy^j$ be two polynomials with coefficients $v_i, w_i$ from a ring $R$. Define 
	\[
  q_1(y) \star q_2(y)
    :=  1+\sum\limits_{i=1}^{d_1}v_iy^i+\sum\limits_{j=1}^{d_2}w_jy^j+\sum\limits_{1 \leq i \leq d_1, 1 \leq j \leq d_2} v_iw_jy^{u(i, j)},
    \]
    where $u(i, j)$ is the integer such that \[(1+\epsilon)^{u(i, j)} \geq (1+\epsilon)^i + (1+\epsilon)^j > (1+\epsilon)^{u(i, j)-1}.\]
\end{definition}

Note the operation $\star$ implicitly depends on the $\epsilon$ chosen. Intuitively, the $\star$ operation uses the exponents of polynomials to approximately count. The exponent represents $\log_{1+\epsilon}(\text{count})$ approximately. If we had $(1+\epsilon)^{u(i, j)} = (1+\epsilon)^i + (1+\epsilon)^j$, where $i = \log_{1+\epsilon}(\text{count}_1)$ and $j = \log_{1+\epsilon}(\text{count}_2)$, then we would in fact have $(1+\epsilon)^{u(i, j)} = \text{count}_1+\text{count}_2$, i.e., our counting would be exact. As we only have $(1+\epsilon)^{u(i, j)} \geq (1+\epsilon)^i + (1+\epsilon)^j > (1+\epsilon)^{u(i, j)-1}$, we potentially lose a multiplicative factor of $(1+\epsilon)$ every time we apply the $\star$ operation. Note that the approximation factor improves, as we decrease $\epsilon$.

Now we give the pseudocode of our deterministic algorithm.
\begin{figure}[H]
	\begin{framed}
		\begin{algorithmic}[1]
			\REQUIRE $\mathtt{Evaluate2}(x, A, t)$
			\STATE Let $L_i:= A\cap (t/2^i, t/2^{i-1}]$ for all $i=1,2,\dots,\lceil \log n\rceil - 1$.
			\STATE Let $L_{\lceil \log n \rceil} := A \backslash (L_1\cup \dots \cup L_{\lceil \log n \rceil-1})$.
			\STATE Set $u := 1$.
			\FOR{$i=1,\ldots, \lceil \log(n)\rceil$}
			\STATE $u := u\cdot \mathtt{ApproxCount}(x, L_i, 2^i)$.
			\ENDFOR
			\STATE \algorithmicreturn\ $u$
		\end{algorithmic}
	\end{framed}
\end{figure}
\begin{figure}[H]
	\begin{framed}
		\begin{algorithmic}[1]
			\REQUIRE $\mathtt{ApproxCount}(x, L, z)$
			\STATE Let $L = \{b_0, b_1, \ldots, b_{m-1}\}$. 
			\STATE Let $F(i, i+1, x) := 1+yx^{b_i}$ where $y$ is a formal variable.
			\STATE Let $\epsilon = 1/\log_2(m)$, and define $F(i, j, x)$ (a univariate polynomial in variable $y$) recursively:\\
			~ $F(i, j, x) := F(i, (i+j)/2, x) \star F((i+j)/2, j, x)$. 	   
			\STATE Let  $v = \text{sum of the coefficients of } y^k \text{ in } F(0, m, x)$  for $0 \leq k \leq 1+\log_{1+\epsilon}(z)+\log{m}$.
			\STATE \algorithmicreturn\ $v$
		\end{algorithmic}
	\end{framed}
\end{figure}

For fixed $x$ and $L$, $F(i, j, x)$ is a univariate polynomial in variable $y$ with coefficients depending on $x$ and $L$.

We will use $x^S$ for $S \subseteq \{0,1,\dots,m-1\}$ and $L = \{b_0, b_1,\ldots, b_{m-1}\}$ to denote $x^{\sum_{i \in S} b_i}$. We omit $L$ from the notation as it will be clear from the context.

The key to our analysis is the following lemma regarding $F(i, j, x)$, which proves that $F(i, j, x)$ can be used to approximately count. In particular, $F(i, j, x)$ contains terms of the form $y^kx^S$; we will ensure that the exponent $k$ approximates $\log_{1+\epsilon} \abs{S}$.

\begin{lemma}\label{lem:counting-prop}
    $F(i, j, x)$ has the following properties:
	\begin{enumerate}
	    \item \label{prop:form} $F(i, j, x) = 1+\sum\limits_{k=1}^{p}\sum\limits_{S \in \mathcal{S}_{i, j, k}}y^kx^S$ where $p \leq 1+\log_{1+\epsilon}(m)+\log_2(m)$,  and for all $S \in  \mathcal{S}_{i, j, k}$ we have $\emptyset \neq S \subseteq [i, j)$.
		\item For every $S, i, j$ such that $\emptyset \neq S \subseteq [i, j)$ there exists a unique integer $k$ such that $S \in \mathcal{S}_{i,j, k}$. Furthermore, $\log_{1+\epsilon}(\abs{S}) \leq k \leq 1+ \log_{1+\epsilon}(\abs{S})+\log_2(j-i)$. \label{prop:unique}
	\end{enumerate}
\end{lemma}

\begin{proof}
	We prove all properties simultaneously by induction on $j-i$. For the base case of $j=i+1$, observe that
	\begin{enumerate}
	    \item For $F(i, i+1, x) := 1+yx^{b_i}$ we have $1 = p \leq 1+\log_{1+\epsilon}(m)+\log_2(m)$ and $\mathcal{S}_{i, i+1, 1} = \{\{i\}\}$.
		\item Only the set $S = \{i\}$ satisfies $\emptyset \neq S \subseteq [i, i+1)$, and $\{i\} \in \mathcal{S}_{i, i+1, 1} =  \{\{i\}\}$, where $0 = \log_{1+\epsilon}(\abs{S}) \leq k = 1 \leq 1+ \log_{1+\epsilon}(\abs{S})+\log_2(j-i) = 1$.
	\end{enumerate}
	
	Let us now move to proving the induction hypothesis. Recall $F(i, j, x)$ is defined to be $F(i, j', x) \star F(j', j, x)$ for $j' = (i+j)/2$. By induction we have that $F(i, j', x) = 1+\sum\limits_{k_1=1}^{p}\sum\limits_{S_1 \in \mathcal{S}_{i, j', k_1}}y^{k_1}x^{S_1}$ and $F(j', j, x) = 1+\sum\limits_{k_2=1}^{p}\sum\limits_{S_2 \in \mathcal{S}_{j', j, k_2}}y^{k_2}x^{S_2}$. Hence we have
	\begin{align*}
&	F(i,j,x) \\
	    = \ & F(i, j', x) \star F(j', j, x)\\
	    = \ & \left(1+\sum\limits_{k_1=1}^{p}\sum\limits_{S_1 \in \mathcal{S}_{i, j', k_1}}y^{k_1}x^{S_1}\right) \star \left(1+\sum\limits_{k_2=1}^{p}\sum\limits_{S_2 \in \mathcal{S}_{j', j, k_2}}y^{k_2}x^{S_2}\right)\\
	    = \ & 1 + \sum\limits_{k_1=1}^{p}\sum\limits_{S_1 \in \mathcal{S}_{i, j', k_1}}y^{k_1}x^{S_1} +\sum\limits_{k_2=1}^{p}\sum\limits_{S_2 \in \mathcal{S}_{j', j, k_2}}y^{k_2}x^{S_2}\\
	    ~\ & ~~~ +\sum\limits_{k_1=1}^{p}\sum\limits_{k_2=1}^{p}\sum\limits_{S_1 \in \mathcal{S}_{i, j', k_1}}\sum\limits_{S_2 \in \mathcal{S}_{j', j, k_2}} y^{u(k_1, k_2)}x^{S_1 \cup S_2},
	\end{align*}
	where the last equality follows because $S_1 \cap S_2 = \emptyset$ (as $S_1 \subseteq [i, j')$ and $S_2 \subseteq [j', j)$).

	Consider a nonempty subset $S \subseteq [i, j)$ and let $S_1 = S \cap [i, j')$ and $S_2 = S \cap [j', j)$.  We will prove the existence of a unique monomial $y^kx^S$ which occurs with coefficient 1 and $\log_{1+\epsilon}(\abs{S}) \leq k \leq 1+ \log_{1+\epsilon}(\abs{S})+\log_2(j-i)$ is satisfied.
	
Our analysis has three separate cases:\medskip

\noindent \textbf{Case 1.} 
		 $S_1 = \emptyset$ and $S = S_2 \neq \emptyset$. By induction (Property~\ref{prop:unique}) there exists a unique $k_2$ such that $y^{k_2}x^{S_2} = y^{k_2}x^{S}$ is a monomial in $F(j', j, x)$. Furthermore by induction (Property~\ref{prop:form}) this monomial occurs with coefficient 1. By Definition~\ref{def:aprod} this monomial gets carried over to $F(i, j, x)$ with coefficient 1. As for all $k'$ every set in $\mathcal{S}_{i,j', k'}$ is non-empty (Property~\ref{prop:form}) there is no other occurrence of the monomial $y^{k_2}x^{S}$ in $F(i, j, x)$. Hence we have $k=k_2$.
		
		By induction (Property~\ref{prop:unique}), we have $\log_{1+\epsilon}(\abs{S_2}) \leq k = k_2 \leq 1+\log_{1+\epsilon}(\abs{S_2})+\log_2((j-i)/2)  < 1+\log_{1+\epsilon}(\abs{S})+\log_2(j-i)$.\medskip
		
\noindent \textbf{Case 2.} 
 $S_2 = \emptyset$ and $S = S_1 \neq \emptyset$. This is entirely symmetric to the case above.\medskip
		
\noindent \textbf{Case 3.} 
  $S_1, S_2 \neq \emptyset$. By induction there exist unique $k_1, k_2$ such that $y^{k_1}x^{S_1}$ is a monomial in $F(i, j', x)$ and $y^{k_2}x^{S_2}$ is a monomial in $F(j', j, x)$. Furthermore by induction both of these monomials have coefficient 1. By Definition~\ref{def:aprod} the only term containing $x^S$ will be $y^{u(k_1, k_2)}x^{S} = y^{k}x^{S}$ occurring with coefficient 1. Hence we have $k = u(k_1, k_2)$.
		
		By induction (Property~\ref{prop:unique}), $\log_{1+\epsilon}(\abs{S_1}) \leq k_1$ and $\log_{1+\epsilon}(\abs{S_2}) \leq k_2$, which imply that
		\begin{align*}
		    (1+\epsilon)^{u(k_1, k_2)} &\geq (1+\epsilon)^{k_1}+(1+\epsilon)^{k_2}\\
		    &\geq \abs{S_1}+\abs{S_2}\\
		    &= \abs{S},
		\end{align*}
		or equivalently, $k = u(k_1, k_2) \geq \log_{1+\epsilon}(\abs{S})$.
		
		By induction (Property~\ref{prop:unique}) we have that 
		\begin{align*}
		    k_1 &\leq 1+\log_{1+\epsilon}(\abs{S_1})+\log_2((j-i)/2)\\
		    &= \log_{1+\epsilon}(\abs{S_1})+\log_2(j-i).
		\end{align*}
		Similarly we have $k_2 \leq \log_{1+\epsilon}(\abs{S_2})+\log_2(j-i)$.
		
		Hence we have  
		\begin{align*}\label{eq:6.2}
		     &(1+\epsilon)^{k_1}+(1+\epsilon)^{k_2}\\ \leq \ & \abs{S_1}(1+\epsilon)^{\log_2(j-i)}+\abs{S_2}(1+\epsilon)^{\log_2(j-i)}\\
		    =\ & \abs{S}(1+\epsilon)^{\log_2(j-i)}\\
		    =  \ & (1+\epsilon)^{\log_{1+\epsilon}(\abs{S})+\log_2(j-i)}.
		\end{align*}
		Hence by Definition~\ref{def:aprod}, 
		\begin{align*}
		    k=u(k_1, k_2) &< 1+\log_{1+\epsilon}((1+\epsilon)^{k_1}+(1+\epsilon)^{k_2})\\
		    &\leq 1+\log_{1+\epsilon}(\abs{S})+\log_2(j-i),
		\end{align*}
		which completes the proof by induction. Finally, since $\abs{S} \leq m$ and $j-i \leq m$, we have that $p \leq 1+\log_{1+\epsilon}(m)+\log_2(m)$.
\hfill
\end{proof}

Using Lemma~\ref{lem:counting-prop}, we can infer useful properties of the polynomial returned by $\mathtt{ApproxCount}$:

\begin{lemma}
\label{lem:approxcount}
$\mathtt{ApproxCount}(x, L, z)$ where $\abs{L} = m$ returns $\sum_{S \in \mathcal{S}} x^S$ such that
	\begin{enumerate}
		\item for every set $S \subseteq L$ with $\abs{S} \leq z$, $S \in \mathcal{S}$. \label{prop:ac-1}
		\item for all $S \in \mathcal{S}$, $\abs{S} \leq O(z)$. \label{prop:ac-2}
	\end{enumerate} 
\end{lemma}

\begin{proof}
Let us first prove Property~\ref{prop:ac-1}. By Lemma~\ref{lem:counting-prop}, for every set $S \subseteq L$ with $\abs{S} \leq z$ there is a unique $k$ such that $x^{S}y^k$ is a monomial in $F(0, m, x)$ and that for this unique $k$, $x^{S}y^k$ occurs with coefficient 1. Furthermore, $k \leq 1+\log_{1+\epsilon}(\abs{S})+\log{m} \leq 1+\log_{1+\epsilon}(z)+\log{m}$. Since we are adding the coefficients of $y^k$ for all such $k$, we get that the monomial $x^S$ occurs in the polynomial returned by $\mathtt{ApproxCount}(x, L, z)$ with coefficient 1.

By Lemma~\ref{lem:counting-prop}, if $x^{S}y^k$ is a monomial in $F(0, m, x)$, then $\log_{1+\epsilon}(\abs{S}) \leq k$ or equivalently $\abs{S} \leq (1+\epsilon)^k$. Since we are restricting $k \leq 1+\log_{1+\epsilon}(z)+\log{m}$, it follows that $\abs{S} \leq z(1+\epsilon)^{1+\log m} \leq zm^{O(\epsilon)}\le O(z)$, where the last inequality follows from $\eps=1/\log_2(m)$.
\hfill
\end{proof}

We now use Lemma~\ref{lem:approxcount} to argue about the polynomial returned by $\mathtt{Evaluate2}$.

\begin{lemma}
\label{lem:evaluate2}
	$\mathtt{Evaluate2}(x, A, t)$ returns $\sum_{S \in \mathcal{S}} x^S$ such that
	\begin{enumerate}
		\item for every set $S \subseteq [n]$ with $\sum_{i \in S} a_i \leq t$, $S \in \mathcal{S}$. \label{prop:e2-1}
		\item for all $S \in \mathcal{S}$, $\sum_{i \in S} a_i \leq O(t\log{n})$. \label{prop:e2-2}
	\end{enumerate} 
\end{lemma}

\begin{proof}
    Consider a set $S \subseteq [n]$ with $\sum_{i \in S} a_i \leq t$. Let $S_i = S \cap L_i$. By the definition of $L_i$, $\abs{S_i} \leq 2^i$. By Lemma~\ref{lem:approxcount}, $\mathtt{ApproxCount}(x, L_i, 2^i)$ has a term $x^{S_i}$ with coefficient $1$. As $\mathtt{Evaluate2}(x, A, t) = \prod_i \mathtt{ApproxCount}(x, L_i, 2^i)$ hence $\mathtt{Evaluate2}(x, A, t)$ will have the term $\prod_i x^{S_i} = x^S$ with coefficient $1$. This proves property~\ref{prop:e2-1}.
	
	By Lemma~\ref{lem:approxcount}, $\mathtt{ApproxCount}(x, L_i, 2^i)$ only has monomials of the form $x^{S_i}$ where $S_i \subseteq L_i$ such that 
	\begin{align*}
	    \sum_{j \in S_i} a_j &\leq (\max_{j \in S_i} a_j) \cdot \abs{S_i}\\
	    &\leq (t/2^{i-1}) \cdot  O(2^i)\\
	    &\le  O(t).
	\end{align*}
	
	Any monomial in $\mathtt{Evaluate2}(x, L, z)$ will have the form $\prod_i x^{S_i} = x^{S}$ where $x^{S_i}$ is a monomial in $\mathtt{ApproxCount}(x, L_i, 2^i)$ and $S = \cup_i S_i$. Property~\ref{prop:e2-2} follows from
	\begin{align*}
	    \sum_{j \in S} a_j &= \sum_{i=1}^{\lceil \log n\rceil} \sum_{j \in S_i} a_j\\
	    &\leq \sum_{i=1}^{\lceil \log n\rceil} O(t)\\
	    &\leq O(t\log n). \qedhere
	\end{align*}
\hfill
\end{proof}

\begin{cor}
\label{cor:evaluate2}
Let the output of $\mathtt{Evaluate2}(x,A, t)$ is a polynomial $P(x)$ where $A = [a_1, a_2, \ldots, a_n]$. Then:
\begin{enumerate}
    \item $P(x)$ is a polynomial of degree at most $d = O(t\log n)$, with non-negative coefficients that are bounded above by $2^{\min\{n, d\log n \}}$. \label{prop:ce2-1}
    \item $P(x)$ contains the monomial $x^t$ iff there exists a $R\subseteq [n]$ be such that $t = \sum_{i \in R}a_i$.  \label{prop:ce2-2}
\end{enumerate}
\end{cor}

\begin{proof}
Let us first prove Property~\ref{prop:ce2-1}. By Lemma~\ref{lem:evaluate2} the degree is bounded by $\max_{S \in \mathcal{S}} \sum_{i \in S} a_i \leq O(t\log{n})$.  By Lemma~\ref{lem:evaluate2}, $\mathtt{Evaluate2}(x,A, t)$ is a sum of monomials of the form $x^{S}$ where $S \subseteq [n]$ and the coefficient of $x^S$ is either $0$ or $1$. Hence for any $a \in \mathbb{Z}_{\geq 0}$ the coefficient of $x^{a}$ is non-negative and bounded above by the number of subsets with sum equal to $a$ which is always at most $ \binom{n}{a}\le \min\{2^n, n^a\}$ as all $a_i \in \mathbb{Z}_{> 0}$. Using $a \leq d$ the bound on the coefficients follows.

As all monomials are of the form $x^S$ for $S \subseteq [n]$ we can only have the monomial $x^t$ if there exists a set $R\subseteq [n]$ be such that $t = \sum_{i \in R}a_i$.
Conversely, if there exists a set $R\subseteq [n]$ be such that $t = \sum_{a\in R}a$ then $P(x)$ contains the monomial $x^t$ by  Property~\ref{prop:e2-2} of Lemma~\ref{lem:evaluate2}.
\hfill
\end{proof}

\subsection{Implementation}
Now we describe in more detail how to implement the procedures $\mathtt{Evaluate2}$ and $\mathtt{ApproxCount}$ with low time and space.
 
\begin{lemma}
\label{lem:implement2}
The procedure $\mathtt{Evaluate2}$ (where arithmetic operations are over $\F_q$ with $q=\Omega(t)$) can be implemented in $O(n \cdot \polylog(qn) )$ time and $O(\log{q} \cdot \log^3{n})$ working space.
 \end{lemma}
 
 \begin{proof}
Recall $A = \{a_1,\ldots,a_n\}$ is the set of input integers. In $\mathtt{Evaluate2}(x,A,t)$, we do not have enough space to collect all elements of $L_i$ and pass them to $\mathtt{ApproxCount}$. Instead, we will pass the list $A$ (i.e., our input). To correct this, in $\mathtt{ApproxCount}(x, L, z)$ when we encounter an $a_i \not\in L$ we just ignore it by setting $F(i, i+1, x) = 1$. 

By Property~\ref{prop:form} of Lemma~\ref{lem:counting-prop}, any polynomial of the form $F(i, j, x)$ computed in $\mathtt{ApproxCount}(x, L, z)$ has degree at most $1+\log_{1+\epsilon}(m)+\log_2(m)$, where $\eps=1/\log_2(m)$. Thus, the space needed to store a single polynomial is \begin{align*}
O(\log{q}\cdot (1+\log_{1+\epsilon}(m)+\log_2(m))) \leq  O(\log{q}\cdot \log^2 m).
\end{align*}
To compute $F(0, m, x)$ in $\mathtt{ApproxCount}(x, L, z)$, we perform the recursion in a depth-first way, with recursion depth at most $\log n$. At any point in time, we will have only stored (at most) one polynomial at every level of the recursion tree, so
the total space usage of $\mathtt{ApproxCount}(x, L, z)$ is $O(\log q\cdot \log^3 n)$.

In $\mathtt{ApproxCount}(x, L, z)$, we compute the approximate counting product $O(m)$ times, where each multiplication takes $\poly(\log q, \log m,\eps^{-1})$ time. So the  running time of $\mathtt{ApproxCount}(x, L, z)$ is $O(n \cdot \polylog(qn))$.
It follows that the total running time and space requirements of $\mathtt{Evaluate2}(x,A,t)$ are $O(n \cdot \polylog(qn))$ and $O(\log{q} \cdot \log^3{n})$, respectively.
\hfill
\end{proof}

Finally, we can complete the proof of Theorem~\ref{thm:main_det} and give the final deterministic algorithm.

\begin{proofof}{Theorem~\ref{thm:main_det}}
By Lemma~\ref{lem:kane-improved}, we have a deterministic algorithm for checking if the monomial $x^t$ has a nonzero coefficient in $P(x) = \mathtt{Evaluate2}(x,A,t)$, in time $\tilde O((d+w)(T+w)w)$, where:
\begin{itemize}
    \item $d$ denotes the degree of $P(x)$ and $d \leq O(t \log n) $ by Corollary~\ref{cor:evaluate2},
    \item $2^w$ denotes the largest coefficient of $P(x)$ and $w \leq \min\{n, d\log(n)\}$ by  Corollary~\ref{cor:evaluate2}, and
    \item $T$ denotes the time to calculate $P(x)$ for a given $x \in \mathbb{F}_q$ for $q \leq  O(d+w) = \tilde O(t \cdot\polylog n )$.
\end{itemize}
By Lemma~\ref{lem:implement2}, we have $T \leq O(n \cdot \polylog(qn)) \leq \tilde O(n \cdot \polylog(nt))$. Plugging in the upper bounds for $d, T, w$, the total running time is
\begin{align*}
\tilde O((d+w)(T+w)w) &\le \tilde O((d\log n)\cdot (n\polylog (nt))\cdot n)\\     
& \le \tilde O(n^2 t).
\end{align*}

By Lemma~\ref{lem:kane-improved} the space of the deterministic algorithm is $O(S + \log(d+w))$ where $S$ is the space required to calculate $P(x)$ for a given $x \in \mathbb{F}_q$ for $q \leq  O(d+w) \le   \tilde O(t\polylog n)$. By Lemma~\ref{lem:implement2}, $S = O(\log{q} \cdot \log^3{n})$. Hence, assuming $t\ge \log (n)$, we have $\log q = O(\log t)$ and the overall space complexity is $O(\log t\cdot \log^3 n)$.  In the case of $t<\log n$, we simply use the deterministic $O(nt)$-time $O(t+\log n)$-space dynamic programming algorithm instead.
\hfill
\end{proofof}

\section{Approximation Algorithms}
\label{sec:apx}
In this section, we present a fast low-space randomized algorithm for the following {\bf Weak Subset Sum Approximation Problem} (a.k.a. WSSAP):

\begin{definition}[WSSAP]
    Given a list of positive integers $A = [a_1, a_2, \ldots, a_n]$, target $t$, $0 < \eps < 1$ with the promise that they fall into one of the following two cases:
    \begin{itemize}
        \item {\bf YES:} There exists a subset $S \subseteq [n]$ such that $(1-\eps/2) t\le \sum_{i \in S} a_i \le t$.
        \item {\bf NO:} For all subsets $S \subseteq [n]$ either $\sum_{i \in S} a_i > (1+\eps)t$ or $\sum_{i \in S} a_i < (1-\eps)t$.
    \end{itemize}
    decide whether it is a YES instance or a NO instance.
\end{definition}
The search version of the above definition was introduced by Mucha, W\k{e}grzycki, and W{\l}odarczyk~\cite{MuchaW019} as a ``weak'' notion of approximation. They gave a $\tilde{O}(n+1/\eps^{5/3})$ time and space algorithm.

Note the usual decision notion of ``approximate subset sum'' distinguishes between the two cases of

\noindent (1) there is an $S\subseteq [n]$ such that $\sum_{i\in S}a_i \in [(1-\eps/2)t, t]$, and \\
\noindent (2) for all $S\subseteq [n]$, $\sum_{i\in S}a_i <(1-\eps)t$ or $\sum_{i\in S}a_i >t$.\\  
This is, in principle, a harder problem. 

\begin{reminder}{Theorem~\ref{thm:approx}}
There is a $\tilde{O}(\min\{n^2/\eps, n/\eps^2\})$-time and $O(\polylog(n, t))$-space algorithm for WSSAP.
\end{reminder}

\begin{proof} Our algorithm runs two different algorithms, and takes the output of the one that stops first. Algorithm 1 will use $\tilde{O}(n^2/\eps)$ time and $\polylog(nt)$ space; Algorithm 2 will use $\tilde{O}(n/\eps^2)$ and $\polylog(nt)$ space.
\medskip

\noindent \textbf{Algorithm 1:} Define $b_i = \floor{\frac{a_i}{N}}$ where $N = \eps t/(2n)$. First we will prove that $(A, t)$ is a YES instance if and only if there is a subset $S\subseteq[n]$ such that $N\sum_{i \in S} b_i \in [t-\eps t, t]$.

Assume $(A, t)$ is a YES instance. We have $a_i - N \leq N b_i \leq a_i$; hence for set $S \subseteq [n]$ such that $\sum_{i \in S} a_i \in [t(1-\eps/2),t]$, we have $N\sum_{i \in S} b_i \in \left[\sum_{i \in S} a_i - N n,\sum_{i \in S} a_i\right] \subseteq [t(1-\eps/2)-N n, t] = [t-\eps t, t]$.

On the other hand, suppose there is an $S \subseteq [n]$ such that $N\sum_{i \in S} b_i \in [t-\eps t, t]$. Then as $Nb_i \leq a_i \leq Nb_i + N$ we have that $\sum_{i \in S} a_i \in [t-\eps t, t+\eps t/2]$, which implies that the original instance was a YES instance.  

Therefore, $(A, t)$ is a YES instance if and only if there is a set $S$ such that $N\sum_{i \in S} b_i \in [t-\eps t, t]$, i.e., $\sum_{i \in S} b_i \in [2n(1-\eps)/\eps, 2n/\eps] = [t'(1-\eps), t']$ for $t' = 2n/\eps$.

So we have reduced the original problem to a subset sum instance on a list $B = [b_1, b_2, \ldots, b_n]$ where we want to know if there is a subset with sum in the range $[t'(1-\eps), t']$. Our algorithm (Theorem~\ref{thm:main}) can also handle this modification, as all we need to change is that we need to detect if there is a monomial of the form $x^{d}$ for some $d \in [t'(1-\eps), t']$ instead of $d = t'$. This change can be handled in Corollary~\ref{cor:kane} by multiplying with \[\sum_{i=t'(1-\eps)}^{t'} x^{q-1-i} = x^{q-1-t'}(1-x^{\eps t'+1})/(1-x),\] instead of $x^{q-1-t'}$. Alternatively, one could also use the reduction described in Remark~\ref{rem:range}. By Lemma~\ref{lem:kane-improved}, Algorithm 1 runs in $\tilde{O}(nt') = \tilde{O}(n^2/\eps)$ time and $\polylog(nt)$ space.
\medskip

Now we describe Algorithm 2.
\medskip

\noindent \textbf{Algorithm 2:} Let $S_{big} = \{i \mid a_i > \eps t\}$, $A_{big} = \{a_i \mid a_i > \eps t\}$ and $S_{small} = \{i \mid a_i \leq \eps t\}$, $A_{small} = \{a_i \mid a_i \leq \eps t\}$. Let $h = \sum_{i \in S_{small}} a_i$. 

If there is a subset of $A_{big}$ with sum in $[(1-\eps)t-h, (1+\eps)t]$, then we can add elements from $A_{small}$ to the set, until our sum is in the range $[(1-\eps)t, (1+\eps)t]$. Hence the input must be a YES instance.
On the other hand, if some subset of $A$ has a sum in the range $[(1-\eps/2)t, t]$, the restriction of this subset on $A_{big}$ has a sum in the range $[(1-\eps/2)t-h, t]$. 
Therefore, deciding if $(A,t)$ is a YES-instance is equivalent to deciding if there exists a subset of $A_{big}$ with sum in $[(1-\eps)t-h, (1+\eps)t]$. 

As all elements in $A_{big}$ have values greater than $\eps t$, the number of elements in a subset of $A_{big}$ with sum in $[(1-\eps)t-h, (1+\eps)t]$ is at most $(1+\eps)t/(\eps t) \leq 2/\eps$.  

Define $b_i = \floor{\frac{a_i}{N}}$ where $N = \eps^2 t/8$. 

We claim that $(A, t)$ is a YES instance if and only if there exists a set $S\subseteq S_{big}$ such that $N(\sum_{i \in S} b_i) \in [t(1-\eps)-h, (1+\eps/2)t]$.

We have $a_i - N \leq Nb_i \leq a_i$. If $(A, t)$ is a YES instance, then there is a subset $S \subseteq S_{big}$ with $\sum_{i \in S} a_i \in [(1-\eps/2)t-h, t]$ which implies that $N\sum_{i \in S} b_i \in \left[\sum_{i \in S} a_i - N\abs{S}, \sum_{i \in S} a_i\right] \subseteq [(1-\eps/2)t-h-(\eps^2 t/8)(2/\eps), t] \subseteq [t(1-\eps)-h, t]$.

For the other direction, suppose there exists a set $S \subseteq S_{big}$ such that $N\sum_{i \in S} b_i \in [t(1-\eps)-h, t(1+\eps/2)]$. For all $i \in S_{big}$, $b_i \geq (a_i-N)/N \geq (\eps t - \eps^2 t/8)/(\eps^2 t/8) \geq 6/\eps$, and hence $\abs{S} \leq \big (t(1+\eps/2)/N\big )\big /(6/\eps) =  2(2+\eps)/(3\eps) \leq 2/\eps$. Then as $Nb_i \leq a_i \leq Nb_i + N$ we have that $\sum_{i \in S} a_i \in [t(1-\eps)-h, t(1+\eps/2)+N\abs{S}] \subseteq [t(1-\eps)-h, t(1+\eps)]$, which implies that the original instance $(A, t)$ was a YES instance. This completes the proof of our claim above.


We have now reduced the original problem to a list $B_{big} = \{b_i \mid a_i \geq \eps t\}$ where we want to know if there is a subset with sum in the range $[t'(1-\eps/2)-h/N, t'(1+\eps/2)]$ for $t' =t/N= 8/\eps^2$. As in Algorithm 1, we can solve this in time $\tilde{O}(nt'(1+\eps/2)) \leq \tilde{O}(n/\eps^2)$ and $\polylog(nt)$ space.

Combining Algorithms 1 and 2, we obtain an algorithm running in $\tilde{O}(\min\{n^2/\eps, n/\eps^2\})$ time and $O(\polylog(nt))$ space.
\hfill
\end{proof}

	\section{Conclusion}

    In this paper, we have given novel \SS algorithms with about the same running time as Bellman's classic $O(nt)$ time algorithm, but with radically lower space complexity. We have also provided algorithms giving a general time-space tradeoff for \SS. The algorithms apply several interesting number-theoretic and algebraic tricks; we believe these tricks ought to have further applications.

    We conclude with some open problems. First, the fastest known pseudopolynomial algorithm for \SS runs in $\tilde{O}(n+t)$-time \cite{karl,jw19}. 
	The fastest known $O(\poly\log(nt))$-space algorithm is given in this work.  Is there an algorithm running in $\tilde{O}(n+t)$ time and $\poly(n, \log t)$ space? Perhaps some kind of conditional lower bound is possible, but this may require a new kind of fine-grained hypothesis. 
	
	Second, can the time-space tradeoff in our algorithm (Theorem~\ref{thm:main_tradeoff}) be improved to work all the way to $\poly\log(nt)$ space? Currently it only works down to $\tilde{O}(t/\min\{n,t\})$ space. 
	
	Finally, another interesting open problem is whether our efficiently invertible hash families have further applications. They should be particularly useful for constructing randomized algorithms using low space.
	
	
	

\section*{Acknowledgements}
Ce Jin would like to thank Jun Su and Ruixiang Zhang for pointing him to the Bombieri-Vinogradov theorem.

\bibliographystyle{alphaurl}
\bibliography{ref}

\appendix
\section{Proof of the load-balancing property}
\label{apx:balance}

We will utilize a tail bound for $2k$-wise $\delta$-dependent random variables.

\begin{lemma}[{\cite[Lemma 2.2]{balls}}, {\cite[Lemma 2.2]{br94}}] \label{tail-lemma}
Let $X_1, \dots , X_n \in \{0, 1\}$ be $2k$-wise $\delta$-dependent random variables, for some $k \in \N$
and $0 \le \delta < 1$, and let $X = \sum_{i=1}^n  X_i$ and $\mu = \Ex [X]$. Then, for any $t > 0$ it holds that
$$\Pr [|X - \mu| > t] \le 2\left (\frac{2nk}{t^2}\right )^k + \delta  \left (\frac{n}{t}\right )^{2k}.$$
\end{lemma}

Now we analyze the load-balancing guarantee of this construction as in \cite{balls}. We follow the discussion in Remark~\ref{rem:tree} and view the evaluation of $h(x)$ as tracing the tree path along which $x$ moves from root to leaf.
For $i\in [d]$, recall that   $g_i\colon [n_i]\to [2^{\ell_i}]$ is sampled from a  $k_i$-wise $\delta$-dependent family, 
and bijection $f_i\colon[2^{\ell_i}] \times [n_i] \to [n_{i-1}]$ is defined by \[f_i(b,u) = (b \oplus g_i(u))\circ u.\] 
For a node $B$ at level $(i-1)$,
the $u$-th element of its $b$-th child equals the $f_i(b,u)$-th element of array $B$. That is, the $s$-th element of array $B$ is assigned to its $h_i(s)$-th child, where $h_i\colon[n_{i-1}]\to [2^{\ell_i}]$ is defined by
\[s = b_s\circ u_s,\, h_i(s):= b_s \oplus g_i(u_s).\]
\begin{lemma}[Similar to Lemma 3.2 in \cite{balls}]
\label{lem:3.2}
For any $i\in \{0,1,\dots,d-2\}$, $\alpha \geq \Omega(1/\log \log n)$, $0<\alpha_i < 1$, and set $S_i \subseteq[n_i]$ of size at most $(1+\alpha_i)m_i$, 
$$\Pr_{h_i} \Big [\max_{y\in \{0,1\}^{\ell_{i+1}}} |h^{-1}_{i+1}(y) \cap S_i| \le (1+\alpha)(1+\alpha_i)m_{i+1}\Big]$$ is at least $1-\frac{1}{n^{c+1}}$.
\end{lemma}
\begin{proof}
Fix $y \in \{0,1\}^{\ell_{i+1}}$, let $X = |h_{i+1}^{-1}(y) \cap S_i|$. Without loss of generality, we assume $|S_i| \ge \lfloor (1+\alpha_i)m_i\rfloor$ (otherwise we could add dummy elements).

Each element from $S_i$ can be expressed as $b_s \circ u_s$, where $b_s \in [2^{\ell_{i+1}}], u_s\in [n_{i+1}]$, and $h_{i+1}(b_s\circ u_s) = b_s \oplus g_{i+1}(u_s)$.
Group $S_i$'s elements according to $u_s$.
Then each group has at most one element such that $h(b_s\circ u_s) = y$. 
Assign each group a random variable from $\{0,1\}$, indicating whether it contains an element being hashed to $y$.
Then $X$ equals the sum of these $k_{i+1}$-wise $\delta$-dependent (since $g_{i+1}$ is) random variables. And $\Ex[X] = |S_i|/2^{\ell_{i+1}}$. Then by the tail lemma (Lemma~\ref{tail-lemma}) we have 
\begin{align*}
 \Pr [X>(1+\alpha) \mu]
&\le 2\left (\frac{|S_i|k_{i+1}}{(\alpha \mu)^2}\right) ^{k_{i+1}/2} + \delta \left( \frac{|S_i|}{\alpha \mu} \right )^{k_{i+1}}\\
& = 2\left ( \frac{2^{2\ell_{i+1}} k_{i+1}}{\alpha^2 |S_i|} \right )^{k_{i+1}/2} + \delta \left ( \frac{2^{\ell_{i+1}}}{\alpha}\right )^{k_{i+1}}.
\end{align*}
Since $|S_i|\ge m_i \ge 2^{4\ell_{i+1}}$ and $\alpha = \Omega(1/\log \log n)$, the first summand
\begin{equation}
\label{eqn:31}
    2\left ( \frac{2^{2\ell_{i+1}} k_{i+1}}{\alpha^2 |S_i|} \right )^{k_{i+1}/2} \le 2\left ( \frac{k_{i+1}}{\alpha^2 2^{2\ell_{i+1}}} \right )^{k_{i+1}/2} \le  \frac{1}{n^{c+2}},
\end{equation}
where the last inequality follows from the choice of $k_{i+1}$ and $\ell_{i+1}$ such that $k_{i+1}\ell_{i+1} = \Omega(\log n)$. This also enables us to upper bound the second summand, noting that for an appropriate choice of $\delta = \poly(1/n)$ it holds that 
\begin{equation}
    \delta \left ( \frac{2^{\ell_{i+1}}}{\alpha}\right )\le \frac{1}{2n^{c+2}}.\label{eqn:32}
\end{equation}

Therefore, by combining Equations~(\ref{eqn:31}) and (\ref{eqn:32}), and recalling that $m_{i+1}=m_i/2^{\ell_{i+1}}$ we obtain
\begin{align*}
     \Pr[X>(1+\alpha)(1+\alpha_i)m_{i+1}] &=
    \Pr[X>(1+\alpha)(1+\alpha_i)\frac{m_i}{2^{\ell_{i+1}}}]\\
    &\le \Pr[X>(1+\alpha)\mu]\\
    & \le \frac{1}{n^{c+2}}.    
\end{align*}

The lemma now follows by a union bound over all $y \in \{0,1\}^{\ell_{i+1}}$; note there are at most $n$ such values.
\hfill
\end{proof}

The rest of the proof follows in almost the same way as in \cite{balls}.\footnote{Since we only need to prove a load-balancing parameter of $O(\log n)$, we omit the last step of the proof from \cite{balls} which aims for proving a stronger $O(\log n/\log \log n)$ bound.}

\begin{proofof}{Theorem~\ref{thm:bijloglog}} Fix a set $S\subseteq [n]$ of size $m$, and let $\alpha = \Omega(1/\log\log n)$.
We inductively argue that for every level $i\in \{0,1,\dots, d-1\}$, with probability at least $1-i/n^{c+1}$ the maximal load in level $i$ is at most $(1+\alpha)^i m_i$ elements per bin. For $i=0$ this follows from definition.
For inductive step, we assume the claim holds for level $i$. Now we apply Lemma~\ref{lem:3.2} for each bin in level $i$ with $(1+\alpha_i) = (1+\alpha)^i$.
By union bound, with probability at least $1-(i/n^{c+1}+1/n^{c+1})$, the maximal load in level $i+1$ is at most $(1+\alpha)^{i+1}m_{i+1}$, which shows the inductive step.
In particular, this guarantees that with probability at least $1-(d-1)/n^{c+1}$, the maximal load in level $d-1$ is $(1+\alpha)^{d-1} m_{d-1}\le 2m_{d-1}$, for some appropriate choice of $d=O(\log \log n)$.

In order to bound the number $m_{d-1}$, we note that for every $i\in [d-1]$ it holds that $\ell_i\ge (\log m_{i-1})/4-1$, so $m_i = m_{i-1}/2^{\ell_i} \le 2m_{i-1}^{3/4}$. 
By induction, we have \[m_i \le 2^{\sum_{j=0}^{i-1}(3/4)^j}n^{(3/4)^i} \le 16n^{(3/4)^i}.\]
Thus for an appropriate choice of $d=O(\log \log n)$ it holds that $m_{d-1} \le \log n$. The proof directly follows.
\hfill
\end{proofof}

\begin{proofof}{Corollary~\ref{cor:hasheps}} (Sketch)
Let $\eps>0$ be given. We simply modify the previous construction by reducing the depth $d$ to a (sufficiently large) constant.
Then, as in the proof of Theorem~\ref{thm:bijloglog}, for an appropriate choice of $d=O(1)$, we have $m_{d-1}\le n^{\eps}$. 
The rest of the proof is the same as in Theorem~\ref{thm:bijloglog}.
\hfill
\end{proofof}

\section{Proof of Lemma~\ref{lem:bv}}
\label{sec:bv}
We use the Bombieri-Vinogradov theorem \cite{bombieri_1965, vin} from analytic number theory.
\begin{theorem}[Bombieri-Vinogradov]
Fixing $A>0$, there is a constant $C>0$ such that for all $x\ge 2$ and $Q\in [x^{1/2}\log^{-A} x, x^{1/2}]$,
$$\sum_{q\le Q}\max_{y\le x}\max_{\begin{smallmatrix}\scriptstyle 1\le a\le q,\\ \scriptstyle \gcd(a,q)=1\end{smallmatrix}} \left \lvert \psi(y;q,a)-\frac{y}{\phi(q)}\right \rvert \le Cx^{1/2}Q \log^5 x.$$
Here $\phi(q)$ is the Euler totient function\footnote{$\phi(q)$ is the number of integers $1\le a\le q$ that are coprime with $q$.}, and
$$\psi(y;q,a) := \sum_{\begin{smallmatrix}\scriptstyle n\le y\\ \scriptstyle n\equiv a\pmod{q}\end{smallmatrix}} \Lambda(n),$$
where $\Lambda(n)$ denotes the von Mangoldt function\footnote{$\Lambda(n)= \ln p$ if $n=p^k$ for a prime $p$ and a positive integer $k$; otherwise $\Lambda(n) = 0$}.
\end{theorem}

\begin{lemma}
\label{lem:nttemp}
For sufficiently large $R$ and  $4\le S\le R^{1/2}\log^{-7}R$, there exists $\Omega(R/\log^2 R)$ many prime powers $r=p^k \in (R/2,R]$ such that $r-1$ has an integer divisor in interval $[S/2,S]$.
\end{lemma}
\begin{proof}
Applying the Bombieri-Vinogradov theorem with $A:=7, x:=R, Q := R^{1/2}\log^{-7}R, a:=1$, we have $$\sum_{q\le R^{1/2}\log^{-7}R}\max_{y\le R} \left \lvert \psi(y;q,1)-\frac{y}{\phi(q)}\right \rvert \le CR/\log^2 R.$$

As $S\le R^{1/2}\log^{-7}R$, we can restrict $q$ to the primes in the interval $[S/2,S]$.
By the prime number theorem and the Bertrand-Chebyshev theorem, the number of primes in $[S/2,S]$ is at least $C'\cdot S/\log S$ for some positive constant $C'>0$.
Therefore there are at least $\frac{2C'}{3}\cdot S/\log S$ primes $q\in [S/2,S]$ such that 
$$\max_{y\le R} \left \lvert \psi(y;q,1)-\frac{y}{\phi(q)}\right \rvert \le \frac{CR /\log^2 R}{\frac{C'}{3}\cdot S/\log S} < \frac{3CR /\log R}{C'S}.$$
Setting $y=R/2$ and $y=R$, we get
\begin{align*}
   \psi(R;q,1)-\psi(R/2;q,1)
&= \frac{R/2}{\phi(q)} + \left (\psi(R;q,1)-\frac{R}{\phi(q)}\right )
- \left (\psi(R/2;q,1)-\frac{R/2}{\phi(q)}\right )\\
    &\ge \frac{R/2}{S} -  2\cdot\frac{3CR /\log R}{C'S}\\
    &\ge \frac{R}{4S},
\end{align*}
for sufficiently large $R$. Hence there are at least $R/(4S)$ many $r\in (R/2, R]$ such that $r\equiv 1\pmod{q}$ and $\Lambda(r)>0$, i.e., $r$ is a prime power.

Since the number of choices for prime $q$ is at least $\frac{2C'}{3}\cdot S/\log S$, there exist at least $(\frac{2C'}{3}\cdot S/\log S)\cdot R/(4S) = (C'/6)\cdot R/\log S$ pairs of such $(q,r)$.
As each $r-1$ has at most $\log(r-1)$ prime factors, the the number of distinct $r$ is at least $\Omega(R/(\log S\log R))$.
\hfill
\end{proof}

\begin{cor}
For sufficiently large $R$ and  $4\le S\le R/4$, there exists $\Omega(R/\log^2 R)$ many prime powers $r=p^k \in (R/2,R]$ such that $r-1$ has an integer divisor in interval $[S/2,2S\cdot \log^{15}S]$.
\end{cor}

\begin{proof}
Let $S_0 := R^{1/2}\log^{-7}R$. The case $S\in [4, S_0]$ follows from Lemma~\ref{lem:nttemp}.
Now we assume $S\in (S_0,R/8]$. 
Note that $r-1$ has a divisor in $(S/2,2S\cdot \log^{15} S]$ iff it has a divisor in interval $$\left [\frac{r-1}{ 2S\cdot \log^{15} S},\frac{r-1}{S/2}\right ].$$ For $r\in (R/2,R]$ we have $(r-1)/(S/2) \ge R/S$ and $(r-1)/(2S\cdot \log^{15} S) \le R/(2S\cdot \log^{15} S)$, so it suffices to let $r-1$ have a divisor in interval 
$$I:=\left [\frac{R}{2S\log^{15} S}, \frac{R}{S}\right].$$
Note that  $R/(2S)\ge 4$, and
$$\frac{R}{2S\log^{15} S} < \frac{R}{2S_0\log^{15} S_0} < S_0/2$$
for sufficiently large $R$. So there exists $4\le S'\le S_0$ such that $[S'/2,S']\subseteq I$. We then apply Lemma~\ref{lem:nttemp} with $S'$ in place of $S$.
\hfill
\end{proof}

\end{document}